\documentclass[11pt]{article}

\usepackage{epsfig,epstopdf,multirow,rotating,times,setspace}
\usepackage{graphicx}
\usepackage{amsthm,amsmath,calligra}
\usepackage{amssymb}
\usepackage{booktabs}
\usepackage[all]{xy}
\usepackage{mathtools}
\usepackage{color}
\usepackage{algorithm}
\usepackage{natbib}
\usepackage{mathrsfs}
\usepackage[margin=1in]{geometry}
\def\argmax{\mathop{\rm argmax}}

\newcommand{\s}{\ensuremath{\mathbb{S}}}

\newcommand{\real}{\ensuremath{\mathbb{R}}}

\newcommand{\ltwo}{\ensuremath{\mathbb{L}^2}}
\newcommand{\linf}{\ensuremath{\mathbb{L}^{\infty}}}
\newcommand{\lone}{\ensuremath{\mathbb{L}^1}}

\newcommand{\inner}[2]{\left\langle #1,#2 \right\rangle}

\newtheorem{theorem}{Theorem}
\newtheorem{lem}{Lemma}
\newtheorem{rem}{Remark}
\newtheorem{cor}{Corollary}

\newcommand{\abs}[1]{\left\vert#1\right\vert}
\newcommand{\set}[1]{\left\{#1\right\}}

\newcommand{\norm}[1]{\left\Vert#1\right\Vert}

\usepackage{psfrag}
\usepackage{mathrsfs}
\usepackage{accents}
\usepackage{authblk}
\usepackage{blindtext}

\newcommand{\ptrue}{p_{0}}
\newcommand{\ptemplate}{q}
\newcommand{\pdfs}{{\mathcal P}}
\newcommand{\pdfsm}{\pdfs_{M}}
\newcommand{\pdfsml}{\pdfs_{M, \lambda}}

\usepackage{amsfonts}
\usepackage{amsbsy}
\usepackage{amscd}
\usepackage{amsgen}
\usepackage{amsopn}
\usepackage{amstext}
\usepackage{amsxtra}
\usepackage{cmmib57}
\usepackage[mathscr]{eucal}
\usepackage{color}

\newcommand{\noi}{\noindent}

\newcommand{\laspace}{\;}

\newcommand{\eqcomma}{\laspace,}

\title{{\bf Shape-Constrained Univariate Density Estimation }}

\author[1]{Sutanoy Dasgupta\thanks{Corresponding Author: s.dasgupta@stat.fsu.edu}}
\author[2]{Debdeep Pati}
\author[2]{Ian H. Jermyn}
\author[1]{Anuj Srivastava}
\affil[1]{Department of Statistics, Florida State University, FL, USA}
\affil[2]{Department of Mathematics and Statistics, Durham University, UK}
\affil[3]{Department of Statistics, Texas A\&M University, TX, USA}

\date{}
\begin{document}
\maketitle
\begin{abstract}

While the problem of estimating a probability density function (pdf) from its observations is classical,
the estimation under additional shape constraints is both important and challenging. 
We introduce an efficient, geometric approach for estimating pdfs given the number of its modes. This approach explores 
the space of constrained pdf's using an action of the diffeomorphism group that preserves  their shapes. 
It starts with an initial template, with the desired number of modes and arbitrarily chosen heights at the critical points, 
and transforms it via: (1) composition by diffeomorphisms and (2) normalization to obtain the final density estimate. 
The search for optimal diffeomorphism is performed under the maximum-likelihood criterion and is accomplished by 
mapping diffeomorphisms to the tangent space of a Hilbert sphere, a vector space whose elements 
can be expressed using an orthogonal basis. This framework is first applied to shape-constrained univariate, 
unconditional pdf estimation and then extended to conditional pdf estimation. 
We derive asymptotic convergence rates of the estimator and demonstrate this approach using 
a synthetic dataset involving speed distribution for different traffic flow on Californian driveways. 

\end{abstract}
\tableofcontents{}

\newpage

\section{Introduction}
\label{sec:introduction}

Estimation of a probability density function (pdf) from a number of samples is an important and well-studied problem in statistics. It is useful for any number of statistical analyses, for example, quantile regression, or for indicating data features like skewness or multimodality.  The problem becomes more challenging, however, when additional constraints are imposed on the estimate, especially constraints on the shape of the densities allowed. The imposition of such constraints is motivated by the fact that if the true density is known to lie in a certain shape class, then one should be able to use that knowledge to improve estimation accuracy.

The most commonly studied shape constraints include log-concavity, monotonicity, and unimodality. The obvious extension to multimodality has been studied in the case of function estimation (see the very recent article by \citet{wheeler2017bayesian} and references therein), but there is considerably less work on density estimation under this type of constraint. It is this type of constraint that we focus on in this paper. 

The earliest estimate for a unimodal density was given by \citet{Grenander1956-nt}, who showed that a particular, natural class of estimators for unimodal densities is not consistent, and presented a modification that is consistent. Over the last several decades, a large number of papers have been written analyzing the properties of the \emph{Grenander estimator}, e.g.~\citep{Rao1969-fc,Izenman1991-nk} and its modifications~\citep{Birge1997-gh}. An estimator using a maximum likelihood approach was developed by \cite{Wegman1970-ux}.

The earlier papers  assumed knowledge of the position and value  of the mode, and applied monotonic estimators over subintervals on either side of it. Later papers, for example \cite{Meyer2001-zh,Bickel1996-oi}, include an additional mode-estimation step. Other papers developed Bayesian methods, for example \cite{Brunner1989-rj}. \cite{Hall2002-lw} uses a tilting approach to transform the estimated pdf into the correct shape. \cite{Turnbull2014-ag}, in addition to describing an estimator that uses Bernstein polynomials with the weights chosen to satisfy the unimodality constraint, also provide a useful summary of recent results on unimodal density estimation.

Closer to our work in this paper, \cite{Cheng1999-up} use a \emph{template} function to estimate unimodal densities. Given an unconstrained  estimator, they start from a template unimodal density and provide a sequence of transformations that when applied to the template both keep the result unimodal, and ``improve" the  estimate in some sense. However, the method is \emph{ad hoc}, and convergence, although seen empirically, is by no means guaranteed. 

In contrast, we take a principled geometric approach to the problem. The advantages of our method are as follows. First, while estimation is still based on transformation of an initial template, we apply only a single transformation rather than a (possibly non-convergent) sequence. Coupled with a small number of other parameters, this transformation constitutes a parametrization of the whole of the shape class of interest; there are no hidden constraints. Second, we use a broader notion of shape than previous work: in its simplest form we constrain the pdf to possess a fixed, but arbitrary, number of modes; we consider more general cases in section~\ref{sec:extensiontomoregeneralshapes}. Third, we use maximum likelihood estimation, guaranteeing optimality in principle, and allowing the derivation of asymptotic rates of convergence to the true density.

\subsection{Summary of method}

Our problem can be stated as follows: given independent samples $\set{x_{i}}_{i\in [1..n]}$, from a pdf $\ptrue$, with a known number $M > 0$ of well-defined modes, estimate this density ensuring the presence of $M$ modes in the solution. In order to do this, we construct a parameterization of the set of densities with $M$ modes, $\pdfsm$, as follows. Let the critical points of a pdf $p$ with $M$ modes be $\set{b_{a}}_{a\in [0..2M]}$, with $b_{0} = 0$ and $b_{2M} = 1$. We can define the height ratio vector $lambda$ of $p$ as the set of ratios of the height of the $(a+1)^{\text{th}}$ interior critical point to the height of the first (from the left) mode: $\lambda = \set{\lambda_{1}, \dotsc, \lambda_{2M-2}}$, where $\lambda_{a} = p(b_{a+1})/p(b_{1})$. Let the subspace of $\pdfsm$ with height ratio vector $\lambda$ be denoted $\pdfsml$. We then parameterize an arbitrary member of $\pdfsm$ by the following elements:
\begin{itemize}
\item A height rato vector $\lambda$;

\item a diffeomorphism $\gamma \in \Gamma$, where $\Gamma=\{\gamma: [0,1]\rightarrow [0,1] |\dot{\gamma}> 0,\gamma(0)=0,\gamma(1)=1\}$ is the group of diffeomorphisms of $[0, 1]$.
 
\end{itemize}   

\noindent Together these generate the pdf $p_{\lambda, \gamma} = (\ptemplate_{\lambda}, \gamma) \in \pdfsml$, where $\ptemplate_{\lambda}$ is a \emph{a priori} fixed template function in $\pdfsml$, and $(\cdot, \gamma)$ denotes a group action of $\Gamma$ on $\pdfs$, with the crucial property that it preserves $\lambda$.

Using this parameterization, we can construct the log likelihood function
	\begin{equation}
		L(\set{x_{i}} | \lambda, \gamma)
		=
		\sum_{i} \ln p_{\lambda, \gamma}(x_{i})
		\eqcomma
	\end{equation}

\noi and we can use maximum likelihood to estimate $\lambda$ and $\gamma$. 

The optimization involved is made challenging by the fact that $\Gamma$ is an infinite-dimensional, nonlinear manifold. To address this issue, we define a bijective map from $\Gamma$ into a unit Hilbert sphere (set of square-integrable functions with unit $\ltwo$ norm) and then \emph{flatten} this sphere around a pivot point to reach a proper Hilbert space. Using a truncated orthonormal basis, we can then represent elements of $\Gamma$ by a finite set of coefficients. The joint optimization over $\Gamma$ and $\Lambda$ can then be performed using a standard optimization package since these representations now lie in a finite-dimensional Euclidean space.

We can generalize this method to a larger set of shape classes by defining a shape as a sequence of monotonically increasing, monotonically decreasing, and flat intervals that together constitute the entire density function. For example, the shape of an ``N-shaped'' density function is given by the sequence (increasing, decreasing, increasing). For any such sequence, we can construct a template density in the appropriate shape class, and proceed with estimation as before.

\subsection{Overview}

The rest of the paper is organized as follows. Section~\ref{sec:geometricexploration} describes the parameterization of $\pdfsm$ in detail. Section~\ref{sec:estimationoftheparameters} describes the implementation of the maximum likelihood optimization procedure, and in particular, the parameterization of $\Gamma$. Section~\ref{sec:asymptoticconvergence} presents the asymptotic convergence rates associated with the proposed estimator, while 
Section~\ref{sec:simulationstudy} presents some experimental results on simulated datasets. Section~\ref{sec:extensiontomoregeneralshapes} extends the framework to the more general shape classes just mentioned, while Section~\ref{sec:extensiontoconditional} extends the framework to shape constrained conditional density estimation. Section~\ref{sec:application} presents a application case study. Section~\ref{sec:discussion} summarizes the contributions of the paper and discusses some associated problems, limitations and further possible extensions. The Appendix contains the derivations of the asymptotic convergence rate presented in Section~\ref{sec:asymptoticconvergence}.

\section{A Geometric Exploration of Densities}
\label{sec:geometricexploration}

In this section, describe the parameterization we use for the set $\pdfsm$ of densities with $M > 0$ modes. We start by introducing some notation and some assumptions about the underlying space of densities $\pdfs \supset \pdfsm$. 

In this framework we are primarily going to focus on pdf's that satisfy the following conditions: 
It is strictly positive and continuous with an interval support and zero boundaries.
(For simplicity of presentation, we will assume that the support is $[0,1]$.)
Furthermore, we assume that the pdf has  $M\geq 1$ well defined modes that lie in $(0,1)$.
Let $p$ be such a pdf and suppose that the  $2M+1$ critical points of $p$ are located at $b_i$, for $i=0,\cdots, 2M$, with 
$b_0=0$ and $b_{2M}=1$. Define the {\it height ratio vector} of $p$
to be $\lambda = (\lambda_1, \lambda_2, \dots, \lambda_{2M-2})$, where
 $\lambda_i= p(b_{i+1})/ p(b_1)$ is the ratio of the height of the $(i +1)^{st}$ interior critical point to the height of the first (from the left) mode. Please look at the top left panel of Figure \ref{eg} for an illustration.
We define
${\cal P}$  to be the set of all continuous densities on $[0,1]$ with zero boundaries.
 Let ${\cal P}_M \subset {\cal P}$ be the subset with $M$ modes and let ${\cal P}_{M,\lambda} \subset {\cal P}_M$ be a further subset 
 of pdf's with height ratio vector equal to $\lambda$. 
Define the set of all time warping functions to be $\Gamma=\{\gamma:[0,1] \rightarrow [0,1] | \dot{\gamma} > 0, \gamma (0)= 0, \gamma (1) = 1\}$.
This set is a group with composition being the group operation. The identity element of $\Gamma$ is $\gamma_{id}(t) = t$ and for every 
$\gamma \in \Gamma$ there exists a $\gamma^{-1} \in \Gamma$ such that $\gamma \circ \gamma^{-1} = \gamma_{id}$. 
\begin{theorem}
The group $\Gamma$ acts on the set ${\cal P}_{M,\lambda}$ by the mapping 
${\cal P}_{M,\lambda} \times \Gamma \to {\cal P}_{M,\lambda}$, given by $(p, \gamma)= {p \circ \gamma \over \int (p \circ \gamma)~~ dt}$. 
Furthermore, this action is transitive. That is, for any $p_1, p_2 \in {\cal P}_{M,\lambda}$, there exists a unique $\gamma \in \Gamma$ such that
$p_2 = (p_1, \gamma)$. 
\label{gamexist}
\end{theorem}
\begin{proof}
The new function $\tilde{p} \equiv (p, \gamma)$ is called the {\it time-warped} density or just {\it warped} density. 
To prove this theorem, we first have to establish that the warped density $\tilde{p}$ is indeed in the set ${\cal P}_{M,\lambda}$. 
Note that time warping by $\Gamma$ and the subsequent global scaling do not change the number of modes of $p$ since $\dot{\gamma}$ is strictly positive
(by definition). 
The modes simply get moved to their new locations $\{\tilde{b}_i = \gamma^{-1}(b_i) \}$. 
Secondly, the height ratio vector of $\tilde{p}$ remains the same as that of $p$. 
This is due to the fact that $\tilde{p}(\tilde{b}_i) \propto p(\gamma( \gamma^{-1}(b_i))) = p(b_i)$ 
and $\tilde{\lambda} = \tilde{p}(\tilde{b}_{i+1})/ \tilde{p}(\tilde{b}_1)  = p(b_{i+1})/ p(b_1) = \lambda$. 
Next, we prove the compatibility property that for every $\gamma_1, \gamma_2 \in \Gamma$ and $p$ , we have $(p, \gamma_1 \circ \gamma_2)=((p, \gamma_1), \gamma_2)$. 
Since,  
\begin{eqnarray*}
((p, \gamma_1), \gamma_2)={{p \circ \gamma_1 \over \int (p \circ \gamma_1)~~ ds} \circ \gamma_2 \over \int({p \circ \gamma_1 \over \int (p \circ \gamma_1)~~ ds} \circ \gamma_2)~~ dt}={p \circ \gamma_1 \circ \gamma_2 \over \int (p \circ \gamma_1 \circ \gamma_2)~~ dt}=(p, \gamma_1 \circ \gamma_2)\ ,
\end{eqnarray*}
this property holds. 

Finally, we prove the transitivity property:  
given $p, \tilde{p} \in {\cal P}_{M,\lambda}$, there exists a unique $\gamma_0 \in \Gamma$ such that $\tilde{p}=(p,\gamma_0)$. Let $h_p$ be the height of the first mode of $p$ and let $h_{\tilde{p}}$ be the height of the first mode of $\tilde{p}$. Then, 
define two nonnegative functions according to $g=p/h_p$ and $\tilde{g}=\tilde{p}/h_{\tilde{p}}$. Note that the 
height of both their first modes is $1$ and the height vector for the interior critical points is $\lambda$. 
Also, let the critical points of $p$ and $\tilde{p}$ (and hence $g$ and $\tilde{g}$, respectively) be located at $b_i$ and $\tilde{b}_i$ respectively, for $i=0, \cdots, 2M$.  
Since the modes are well defined, the function $g$ is piecewise strictly-monotonous and 
continuous in the intervals $[b_t,b_{t+1}]$, for $t=0,1,\cdots, 2M-1$. Hence, within each interval $[g(b_t),g(b_{t+1})]$
there exists a continuous inverse of $g$, termed $g_t^{-1}$.
Then, set  $\gamma_1(x)=g_t^{-1}\big(\tilde{g}(x)\big),x \in [\tilde{b}_t,\tilde{b}_{t+1}]$ is such that $(g\circ \gamma_1)=\tilde{g}$ and hence $(p,\gamma_1)=\tilde{p}$. Note that the $\gamma_1$ is uniquely defined, continuous, increasing, but not differentiable at the finitely many critical points $\tilde{b}_i$ in general. Hence $\dot{\gamma_1}$ does not exist at those points. But  $\dot{\gamma_1}$ can be replaced by a weak derivative of $\gamma_1$. Let $D_{\gamma}$ be a weak derivative of $\gamma_1$ that is equal to $\dot{\gamma_1}$ wherever $\dot{\gamma_1}$ exists, and $1$ otherwise. Define $\gamma_0 =\int D_{\gamma}$. Then $\gamma_0$ and $\gamma_1$ are equal and $\dot{\gamma_0}$ exists everywhere, and $(p,\gamma_0)=\tilde{p}$.  
 \end{proof}

Now note that ${\cal P}_M =\underset{\lambda}{\sqcup} {\cal P}_{M,\lambda}$. Thus for $p_0 \in {\cal P}_M$ the estimation procedure entails $(1)$ estimating the (unique) height ratio vector $\lambda_0$ such that $p_0 \in {\cal P}_{M,\lambda_0}$ $(2)$ constructing an element $p_1 \in {\cal P}_{M,\lambda_0}$, and $(3)$ estimating the time warping function $\gamma_0$ such that $p_0=(p_1,\gamma_0)$. Figure \ref{eg2} illustrates the height preserving effect of the composition of warping functions before normalization, and the {\it height ratio vector} preserving effect of the group action.
 
 \begin{figure}[t!]
\begin{center}
\begin{tabular}{cc}
\includegraphics[width=0.44\linewidth]{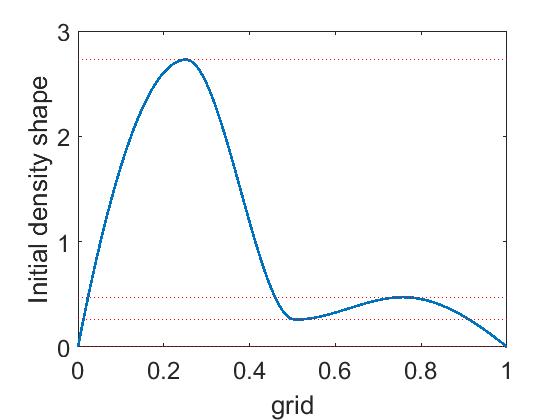} &
\includegraphics[width=0.44\linewidth]{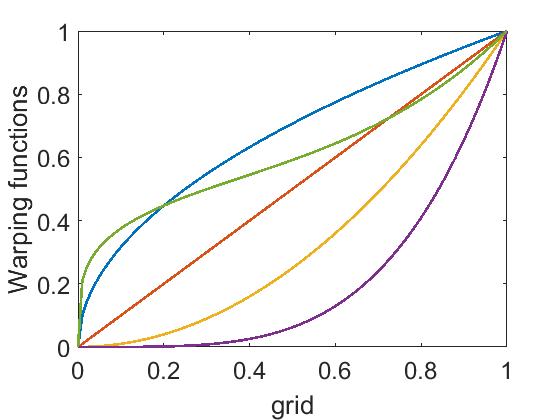}\\
\includegraphics[width=0.44\linewidth]{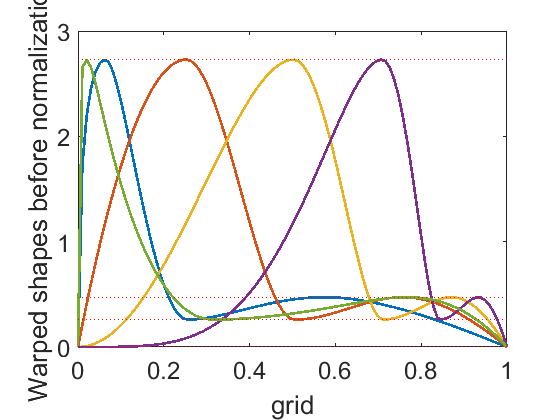} &
\includegraphics[width=0.44\linewidth]{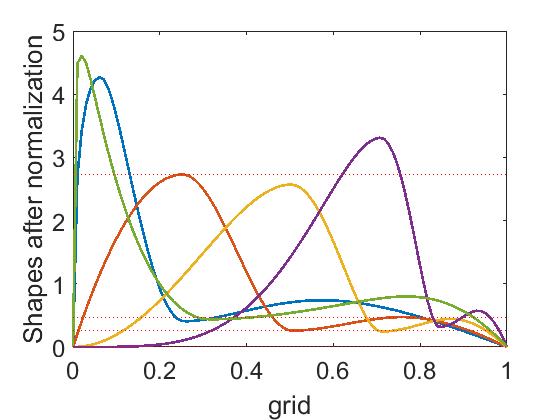}\\
\end{tabular}
\caption{\it The top left panel shows the initial density shape. The top right panel shows different warping functions considered for transforming the shape. The bottom left panel shows the resultant warped shapes which preserves the heights of the critical points. The bottom right panel shows the resultant warped densities after normalization which does not preserve the heights but preserves the {\it height ratio vector}.}
\label{eg2}
\end{center}
\end{figure}

 Assume, for the moment, that $\lambda_0$ corresponding to $p_0$ is known.
The estimation procedure is initialized with an arbitrary $M$ modal template function $g^{\omega}$ constructed as follows:

Set $g^{\omega}(0)=g^{\omega}(1)=\omega$ where $\omega$ is a very small positive number. Let the interval $[0,1]$ be divided into $2M$ equal intervals corresponding to the $M$ modes and $M-1$ interior antimodes. Let the location of the $j$th critical point be $a_j=j/2M$, with $a_0=0$, and $a_{2M}=1$.  Set the value of $g^{\omega}$ for the location of the left most mode $a_1$ to be $1$. Let the heights for the other $2M-2$ interior critical points be $\lambda_i$ for $i=1,\cdots, 2M-2$ which are the height ratio vector for the true density, assumed known for now. Represent this $g^{\omega}$ as $g_{\lambda_0}^{\omega}$. The values of $g_{\lambda_0}^{\omega}$ for the other points is obtained by linear interpolation. Then $ p_1 = g_{\lambda_0}^{\omega} / (\int g_{\lambda_0}^{\omega} ) \in {\cal P}_{M,\lambda}$. The final step of the procedure involves estimating the time warping function $\gamma_0$ such that $(p_1,\gamma_0)=p_0$. 
The key feature of this step of the estimation procedure is the geometry of the set $\Gamma$, which is
crucial in developing a maximum likelihood approach for estimating the $\gamma$ to be used to transform the original shape $g_{\lambda_0}^{\omega}$, since $\Gamma$ is a nonlinear space. Note that $\int_0^1 \dot{\gamma} (u) du=1$. Thus $q=\sqrt{\dot{\gamma}}$ are elements of the Hilbert sphere, with a known simple geometry, and associated linear tangent spaces which facilitate truncated orthogonal expansion to represent the elements of $\Gamma$. The entire procedure of estimating $\gamma_0$ by exploiting the geometry of $\Gamma$ is explained in detail in Section 3. Also the height ratio vector $\lambda_0$, assumed known till now, can be estimated jointly with $\gamma_0$ from the observations via maximum likelihood estimation, discussed in Section 3.2.

Algorithm $1$ provides the steps on how to construct the estimate of $\gamma_0$ given $\lambda_0$ in practice. Note that in practice one can start with a template function rather than a template density because it results in the same estimate of $\gamma_0$.
Figure \ref{eg} shows a simple example to illustrate the estimation procedure. The top left panel is the true density with $M=2$ modes with critical points located at $b_i$ with height $h_i$. The top right panel shows the initial template function with $M=2$ modes and critical points located at $a_i$ and heights $\lambda_i=h_i/h_1$. The bottom left panel shows the warping function constructed according to Algorithm $1$ and the bottom right shows that using the warping function, we get back the exact true density shape. Thus, given any initial template $g_{\lambda}^{\omega}$ the procedure entails estimating the correct height ratios $\lambda_i$ and the warping function.

When the bounds of a density function is not known, they are estimated from the data $X=X_1, X_2, \cdots,X_n$ using the formula $A=\min (X) - sd(X)/\sqrt{n}$ and $B=\max (X) + sd(X)/\sqrt{n}$ where A and B are the lower and upper bounds respectively, $sd(X)$ is the standard deviation of the observations and $n$ is the number of observations, as used in \cite{Turnbull2014-ag}. For a general $A$ and $B$, the data are scaled to the unit interval according to $Z_i=(X_i A)/(B-A)$ for the estimation process.
Note that theoretically the assumption $p_0 (A)=p_0 (B) =0$ can be relaxed  by considering the height ratios of the two boundaries as two extra parameters $\lambda_0$ and $\lambda_{2M+1}$. This allows the proposed framework to encompass a much broader notion of shapes. Specifically, ``shapes" can refer to an ordered sequence of monotonic pieces which when pieced together constitute the entire function. For example, a {\it V shaped} function can be written as a {\it decreasing-increasing} shape. Knowing this ``shape" allows us too incorporate the same shape in the template function and hence obtain a maximum likelihood density estimate in that specified shape class. In fact, this notion even allows one to model know flat modal or antimodal regions in the true density.  However, for experiments, estimating the boundary values to a satisfactory degree requires many observations (using the inbuilt optimization function {\it fmincon}). Hence in this paper we focus on developing the theory for densities which satisfy $p_0 (A)=p_0 (B) =0$. The theory for densities without this assumption is almost identical and results in the same convergence rate, and is not presented. However, we have discussed the idea in more detail in Section $6$ and have also presented some simulated examples. For illustration we focus on densities that decay at the boundaries. Then we estimate the effective support from the data and set the estimate to be zero at the estiated boundaries of the support. In this regard, note that $A$ and $B$ can be any real number and hence the above methodology can be used to estimate densities with entire reals as support. Here $A$ and $B$ play the role of effective support on which the numerical estimation is performed. 

\begin{algorithm}[h]
\caption{Construction of the warping function given a true density $\ptrue$ and the correct critical point height ratios $\lambda_i$ and the critical point locations $b_i$}
i. Start with an M modal template function  $g^{\omega}$. Construct $g^{\omega}$ by setting $g^{\omega}(0)=g^{\omega}(1)=\omega$. Divide the interval $[0,1]$ into $2M$ intervals corresponding to the $M$ modes and $M-1$ antimodes. Let the location of the $j$th critical point be $a_j$, with $a_0=0$, and $a_{2M}=1$. Set the value of $g^{\omega}$ for the first mode to be $1$, that is, $g^{\omega}(a_1)=1$. Let the heights for the other $2M-2$ critical points be the correct height ratios $\lambda_i$ for $i=1,\cdots, 2M-2$ for the true density $p_0$ . Represent this $g^{\omega}$ as $g_{\lambda}^{\omega}$. Obtain the values for the other points by interpolation.\\
ii. Let $\tilde{g}$ be the function $p_0/h_1$. Then $\int_0^1 \tilde{g} dx = 1/h_1$, which implies that $p_0=\tilde{g}/(\int_0^1 \tilde{g} dx)$. Then $\tilde{g}(b_1)=1$ and $\tilde{g}(b_i)=g^{\omega}(a_i)=\lambda_{i-1}$ for $i=2,3,\cdots, 2M-1$. Now,let
\begin{equation}
\Gamma=\{\gamma: [0,1]\rightarrow [0,1] |\dot{\gamma}> 0,\gamma(0)=0,\gamma(1)=1\}
\end{equation}
Then there exists a unique continuous function $\gamma_0$ such that $g_{\lambda}^{0}\circ \gamma_0=\tilde{g}$ where $g_{\lambda}^{0}=g_{\lambda}^{\omega}$ with $\omega=0$. The $\gamma_0$ can be constructed as follows. Since $g_{\lambda}^{0}$ is piecewise monotonous in the intervals $[a_t,a_{t+1}]$, for $t=0,1,\cdots, 2M-1$, there exists an inverse $g_t^{-1}$ in the interval $[g(a_t),g(a_{t+1})]$. Then $\gamma_0$ can be constructed piecewise by
\begin{equation}
\gamma_0(x)=g_t^{-1}\big(\tilde{g}(x)\big),x \in [b_t,b_{t+1}]
\end{equation}
iii. Then $p_0=(g_{\lambda}^{0}\circ \gamma_0)/(\int_0^1 g_{\lambda}^{0}\circ \gamma_0 dx)$
\end{algorithm}

 \begin{figure}[t!]
\begin{center}
\begin{tabular}{cc}
\includegraphics[width=0.44\linewidth]{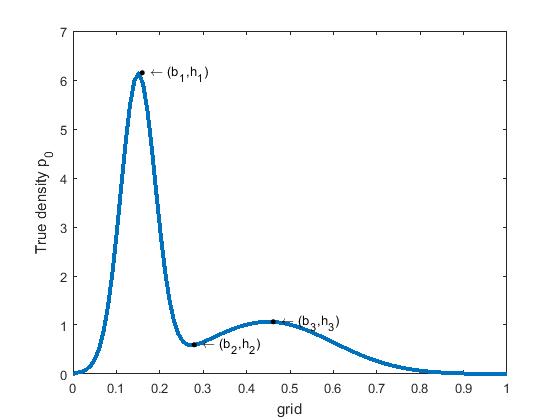} &
\includegraphics[width=0.44\linewidth]{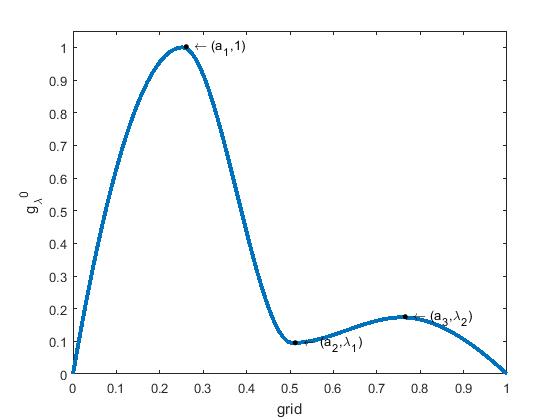}\\
\includegraphics[width=0.44\linewidth]{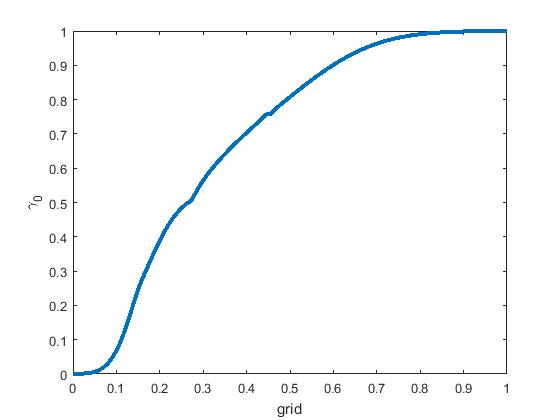} &
\includegraphics[width=0.44\linewidth]{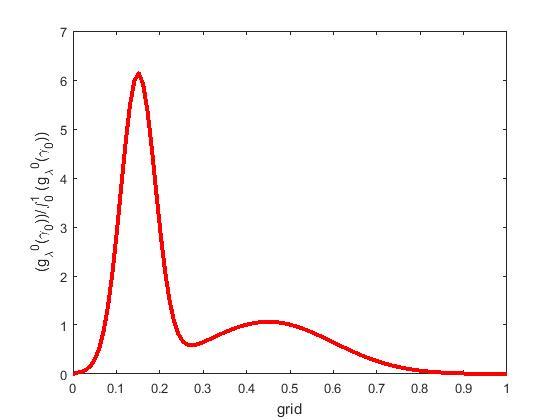}\\
\end{tabular}
\caption{\it The true density(top left) and the constructed template(top right) is shown. The constructed $\gamma_0$ is shown in bottom left panel and the (exact) reconstructed density using the $\gamma_0$ is shown in the bottom right panel.}
\label{eg}
\end{center}
\end{figure}

\section{Estimation of the parameters}
\label{sec:estimationoftheparameters}

In practice, we have to estimate the critical point height ratios $\lambda_i$'s and the warping function $\gamma_0$. We exploit the geometry of the set $\Gamma$ to estimate the desired element $\gamma_0 \in \Gamma$. 

\subsection{Finite-Dimensional Representation of Warping Functions} 
\label{sec:finitedimensionalrepresentation}

Solving an optimization problem, say maximum-likelihood estimation, over $\Gamma$ faces two main challenges. 
First, $\Gamma$ is a nonlinear manifold, and second, it is infinite-dimensional. We handle the nonlinearity by forming 
a map from $\Gamma$ to a tangent space of the unit Hilbert sphere $\s_{\infty}$ (the tangent space is a vector space), 
and infinite dimensionality by 
selecting a finite-dimensional subspace of this tangent space. Together, these two steps are equivalent to finding a  family of finite-dimensional submanifolds $\Gamma_J$ that can be  {\it flattened} into vector spaces. 
This allows for a  representation of $\gamma$ using orthogonal basis. Once we have a finite-dimensional representation of $\gamma$, we 
can optimize over this representation of $\gamma$ using the maximum-likelihood criterion. 

Define a function $q: [0,1] \to \real$, $q(t) = \sqrt{\dot{\gamma}(t)}$, as the square-root slope function (SRSF)
of a $\gamma \in \Gamma$. (For a discussion on SRSFs of general functions, please refer to Chapter 4 of 
\citet{srivastava2016functional}). 
For any $\gamma \in \Gamma$, its SRSF $q$ is an element of the nonnegative orthant of the unit Hilbert sphere, $\s_{\infty} \subset \ltwo$, 
denoted by  $\s_{\infty}^+$. 
This is because
$\|q\|^2 = \int_0^1 q(t)^2 dt = \int_0^1 \dot{\gamma}(t) dt = \gamma(1) - \gamma(0) = 1$.
We have the nonnegative orthant because by definition, $q$ is a nonnegative function. The mapping between $\Gamma$ and 
$\s_{\infty}^+$ is a bijection, with its inverse given by $\gamma(t) = \int_0^t q(s)^2 ds$. 
The unit Hilbert sphere is a smooth manifold with known geometry under the $\ltwo$ Riemannian metric 
\citet{lang2012fundamentals}. It is not a vector space but a manifold with a constant curvature, and can be easily flattened 
into a vector space locally. The chosen vector space is a tangent space of $\s_{\infty}^+$.
A natural choice for reference, to select the tangent space, is the point ${\bf 1} \in \s_{\infty}^+$, 
a constant function with value $1$, which is the SRSF corresponding to $\gamma=\gamma_{\mathrm{id}}(t) = t.$
The tangent space of $\s_{\infty}^{+}$ at ${\bf 1}$
is an infinite-dimensional vector space given by: 
$T_{{\bf 1}}(\s_{\infty}^+) = \{ v \in \ltwo([0,1],\real) | \int_0^1 v(t) dt = \inner{v}{{\bf 1}}= 0\}$.

Next, we define a mapping that takes an arbitrary element of $\s_{\infty}^+$ to this tangent space. For this {\it retraction}, 
we will use the inverse
exponential map that takes $q \in \s_{\infty}^+$ to $T_{{\bf 1}}(\s_{\infty}^+)$  according to:  
\begin{equation}
\exp^{-1}_{{\bf 1}} (q) :\s_{\infty}^+ \to T_{\bf 1}(\s_{\infty}^+),\ \ \ 
v= \exp^{-1}_{{\bf 1}} (q)= {\theta \over \sin(\theta)} (q- {\bf 1} \cos(\theta ))\ ,
 \end{equation} 
 where $\theta =\cos^{-1}(\inner{{\bf 1}}{q})$ is the arc-length from $q$ to ${\bf 1}$. 

We impose a natural Hilbert structure on 
$T_{{\bf 1}}(\s_{\infty}^+)$ using
 the standard inner product: $\inner{v_1}{v_2} = \int_0^1 v_1(t)v_2(t) dt$. 
Further, we can select any orthogonal basis  ${\cal B} = \{ B_j, j=1,2,\dots\}$ of the set $T_{{\bf 1}}(\s_{\infty}^+) $ to express
its elements $v$ by their corresponding coefficients; 
that is,  $v(t) = \sum_{j=1}^{\infty} c_j B_j(t)$, where $c_j= \inner{v}{B_j}$. The only restriction on the basis elements $B_j$'s is that they must be orthogonal to {\bf 1}, that is, $\inner{B_j}{{\bf 1}}=0$. 
In order to map points back from the tangent space to the Hilbert sphere, we use the exponential map, 
given by: 
\begin{equation}
\exp (v) :T_{\bf 1}(\s_{\infty}^+) \to\s_{\infty} ,\ \ \ 
 \exp(v)= \cos(\|v\|){\bf 1} + {\sin(\|v\|) \over \|v\|}\ .
 \end{equation} 
 We define a composite map $H: \Gamma \to \real^J$, as
\begin{equation}
\gamma  \in \Gamma ~~~~~\xrightarrow{\mbox{SRSF}} ~~~~~q = \sqrt{\dot{\gamma}}  \in \s_{\infty}^+ ~~~~~~~
\xrightarrow{\exp^{-1}_{\bf 1}} ~~~~
v 
\in T_{{\bf 1}}(\s_{\infty}^+)  ~~~~ \xrightarrow{ \{B_j\}} ~~~~ 
\{c_j = \inner{v}{B_j} \} \in \real^J \ .
\label{eq:representation}
\end{equation}
 Now, we define $G:\real^J \to \Gamma$, as
\begin{equation}
\{c_j\} \in \real^J \xrightarrow{\{B_j\}} ~ v = \sum_{j=1}^J c_j B_j  \in T_{{\bf 1}}(\s_{\infty}^+)~ \xrightarrow{\exp_{\bf 1}} ~q = 
\exp_{\bf 1}(v) \in \s_{\infty} ~\xrightarrow{~}~   \gamma(t) = \int_0^t q(s)^2 ds\ .
\end{equation}
This map allows us to express an element $\gamma \in \Gamma$ in terms of the coefficient vector $c$. Note that $G$ is not exactly $H^{-1}$ since the range of the exponential map is the entire Hilbert sphere, and not restricted to the nonnegative orthant. We can restrict the domain of $G$ to $V_\pi^J = \{c \in \real^J : \| \sum_{j=1}^{J} c_j B_j\| \leq 2\pi \} \subset \real^J$. Figure \ref{fig:cartoon3} illustrates the map pictorially.

For any $c \in V_{\pi}^J$, let $\gamma_c$ denote the diffeomorphism $G(c)$.
For any fixed $J$, the set $G(V_{\pi}^J)$ is a finite-dimensional submanifold of  $\Gamma$, on which we pose the estimation 
problem.
As $J$ goes to infinity, $G(V_{\pi}^J)$ converges to the set $\Gamma$.

\begin{figure}
\begin{center}
\includegraphics[height=2.78in]{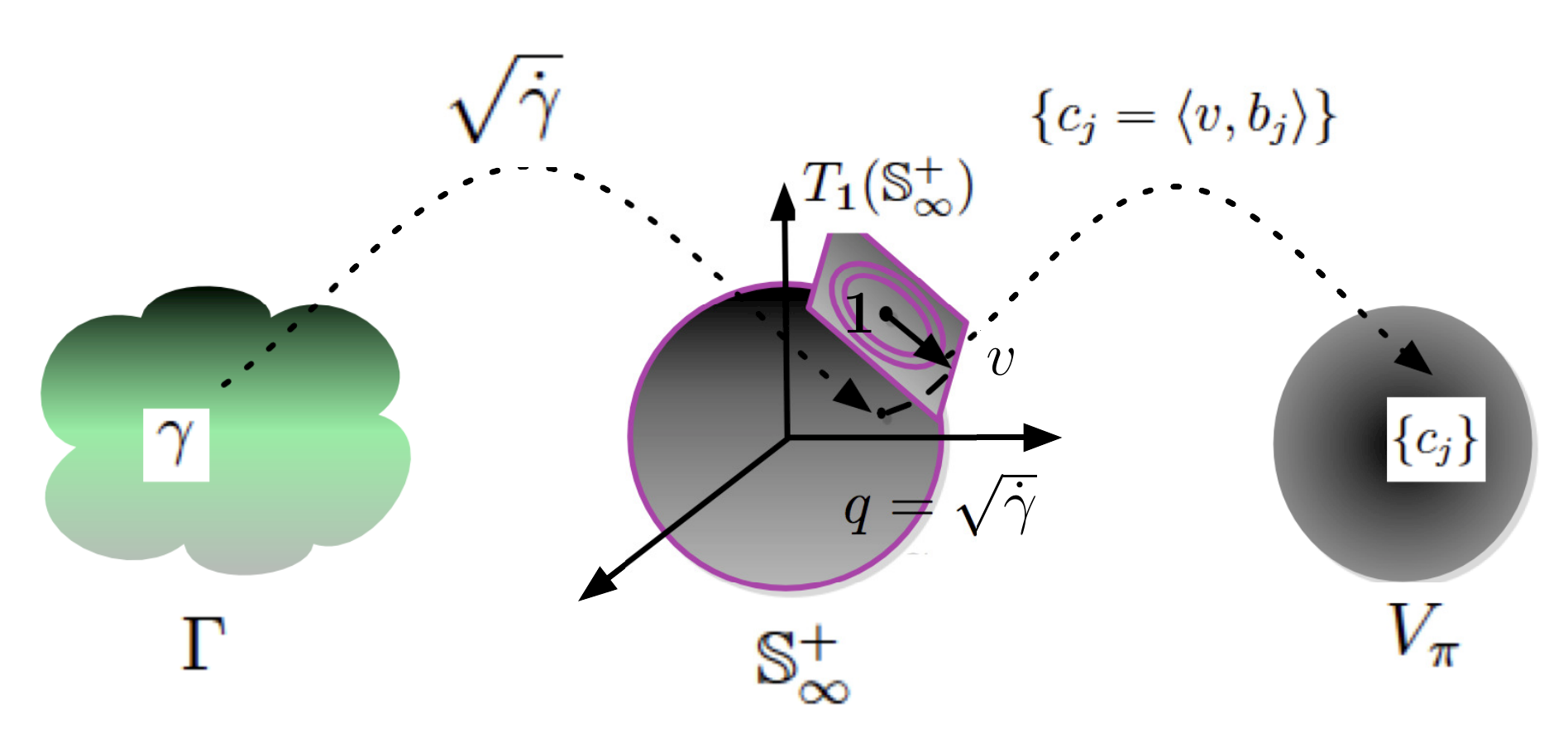}
\caption{A graphic illustration of the finite dimensional representation of elements of $\Gamma$ in terms of the elements $v$ of the tangent space of the Hilbert sphere through the coefficients $c_j$ of the orthogonal basis representation of $v$.}
\label{fig:cartoon3}
\end{center}
\end{figure}

\subsection{Estimation of the $\lambda_i$s and Implementation}
\label{sec:estimationofthelambdai}

We use a joint maximum likelihood method to estimate the height ratios $\lambda_i$s along with the optimal coefficients corresponding to the estimate of $\gamma$. Note that when $M=1$, there is no $\lambda$ parameter. For $M>1$, there are $2M-2$ parameters. Among them, the odd indices $\lambda_1, \lambda_3, \cdots,\lambda_{2M-3}$ correspond to the antimodes, and the rest correspond to the modes. Let $\Lambda=\{\lambda \in \real^{{(2M-2)}^+}| \lambda_1<1, \lambda_1<\lambda_2, \lambda_{2j+1} < \lambda_{2j}, \lambda_{2j+1}<\lambda_{2j+2}, j=1,2,\cdots, M-2\}$.  In the setting described above, the maximum likelihood estimate of the underlying density, given the initial template function $g^{\omega}=g_{\lambda}^{\omega}$, is

$\hat{p}(t)= g_{\hat{\lambda}}^{\omega}(\gamma_{\hat{c}} (t))/(\int_0^1 g_{\hat{\lambda}}^{\omega}(\gamma_{\hat{c}} (t))dt), t\in [0,1]$, where $\gamma_{\hat{c}} = G(\hat{c})$ and 
\begin{equation}
(\hat{c},\hat{\lambda}) =  \argmax_{c \in V_{\pi}^J, \lambda \in \Lambda} \left(  \sum_{i=1}^n \Bigg[ \log \bigg(g_{\lambda}^{\omega} \left(\gamma_c(x_i)\right)/\int_0^1 \left(g_{\lambda}^{\omega} \left(\gamma_c(t)\right)dt\right)\bigg) \Bigg] \right),\ \ 
\gamma_c = G(c)\ .
\label{eq:opt}
\end{equation}

\section{Asymptotic Convergence Results}
\label{sec:asymptoticconvergence}

In this section, we derive the asymptotic convergence rate of the (maximum likelihood) density estimate $\hat{p}$ described according to \eqref{eq:opt} in Section \ref{sec:estimationofthelambdai} to the true underlying density $p_0$ by using the theory of sieve maximum likelihood estimation as in  \citet{wong1995probability}. Let $\mathscr{F}$ denote the set of $M$-modal continuous densities on $[0,1]$ strictly positive in $(0,1)$ and zero at the boundaries.

\begin{itemize}
\item
Assumption 1: $p_0: [0,1] \rightarrow \mathbb{R}^{\geq0}$ is continuous, strictly positive on $(0,1)$, and $p_0(0)=p_0(1)=0$.
\item
Assumption 2: $p_0$ has $M$ modes which lie in $(0,1)$.
\item
Assumption 3: $p_0$ either belongs to H\"{o}lder or Sobolev space of order $\beta$.
\end{itemize}
Let $n$ be the number of available observations. Let $\eta_n$ be a sequence of positive numbers converging to 0. Let $Z_i$ be the of $n$ observed data points scaled to the unit interval. We call an estimator $\hat{p}:[0,1] \rightarrow  \mathscr{F}_n$ an $\eta_n$ sieve MLE if
\[
\frac{1}{n}\sum_{i=1}^{n} \log \hat{p}(Z_i) \geq  \underset{p \in \mathscr{F}_n}{\text{sup}} \frac{1}{n}\sum_{i=1}^{n}\log p(Z_i) -\eta_n
\]
In the proposed method, $\hat{p}$ is defined such that $\frac{1}{n}\sum_{i=1}^{n} \log \hat{p}(Z_i)$   is exactly $ \underset{p \in \mathscr{F}_n}{\text{sup}} \frac{1}{n}\sum_{i=1}^{n}\log p(Z_i)$. 
Therefore, $\hat{p}$ is a sieve MLE with $\eta_n \equiv 0$. Let  $\norm{\cdot}_r$ denote $\mathbb{L}^r$ norm between functions.  The  following theorem states the asymptotic convergence rate for the sieve MLE $\hat{p}$.

\begin{theorem}
Let $\epsilon_n^* = M_1 n^{-\beta/(2\beta+1)} \sqrt{\log n}$ for some constant $M_1$. If $p_0$ satisfies Assumptions 1, 2 and 3; and $\hat{p}$ is the sieve MLE described according to \eqref{eq:opt} in Section \ref{p0}, then there exists constants $C_1$ and $C_2$ such that
\begin{eqnarray}
 P({\|{\hat{p}}^{1/2}-p_{0}^{1/2}\|}_2 \geq \epsilon_n^* ) \leq 5 \exp\big\{-C_2n{(\epsilon_n^* )}^2\big\} + \exp\bigg\{-\frac{1}{4}n C_1{(\epsilon_n^* )}^2\bigg\}.  
\end{eqnarray}
\label{thm2}
\end{theorem}

We present the proof of Theorem \ref{thm2} in Appendix \ref{app:thm2}.  The essential idea hinges on proving the equivalence of the density space $\mathscr{F}$ obtained with the  parameter space. That is, we show that if the estimated parameter is ``close" to the true parameter corresponding to the true density in some sense, then the corresponding estimated density is also ``close" to the true density. The statement is formally stated and proved in Lemma \ref{equiv1} in Appendix \ref{app:thm2}. The general theory is then inspired by the convergence of sieve MLE estimators in \citet{wong1995probability}. 
 
\section{Simulation study}
\label{sec:simulationstudy}

For numerical implementation, we use Fourier basis for the tangent space representation and the MATLAB function {\it fmincon} for optimization. The objective function as described in \eqref{eq:opt} is not convex, and hence the inbuilt function {\it fmincon} is used. However {\it fmincon} can often get stuck in local suboptimal solutions and so we use the {\it GlobalSearch} toolbox along with {\it fmincon} to obtain better results. We start with just $2$ basis points for the tangent space representation and we gradually move towards more number of basis elements upto a predecided limit and choose the estimate based on the best AIC value. AIC was chosen as the penalty on the number of basis elements because experiments suggests that BIC overpenalizes the number of parameter terms which often caused the estimate to miss the sharper features of the true density.


For illustration, we consider sample sizes $100$, $500$ and $1000$. To evaluate the average performances we generate $100$ samples (of sample size $100$, $500$ and $1000$ respectively) and evaluate the mean error and the standard deviation of the errors. For error function we have considered $\ltwo$, $\mathbb{L}^1$ and $\mathbb{L}^{\infty}$. As a first part of the experiment, we generate from three examples with the constraint that the number of modes is one. For comparison, we use the \texttt{umd} packge developed by \cite{Turnbull2014-ag}. In Figure \ref{fig:simulated} we illustrate the best, median and worst performance out of the $100$ samples based on the $\ltwo$ loss function for sample size $100$ for the warped method(top row) and the \texttt{umd} package(bottom row). The examples are as below:

\begin{enumerate}
\item
$p_0=4/5\mathcal{N}(0,4) + 1/5\mathcal{N}(0,0.5)$- a symmetric unimodal example.

\item
$p_0=Beta(9,3)$- a skewed unimodal density with $A=0$, $B=1$ assumed known as well.

\item
$p_0=0.95\mathcal{N}(0,0.5) + 0.05\mathcal{N}(3,1)$-An example of unimodal contaminated data.  

\end{enumerate}

For the symmetric unimodal example, the warped method captures the sharp peak better than the \texttt{umd} method. In the contaminated data example the \texttt{umd} solver gets stuck in a suboptimal solution in one isolated case. However,in the Beta density example the \texttt{umd} method performs better. The quantitative analysis is presented in Table \ref{tablede1}. As a second part of the simulation study we provide two examples with number of modes constrained to be $2$ and $3$ respectively, $(1) p_0=1/3 \mathcal{N}(-1,1) + 2/3\mathcal{N}(1,0.3)$-an asymmetric bimodal example, and $(2) p_0=1/3\mathcal{N}(-1,0.25) + 1/3 \mathcal{N}(0,0.25) + 1/3 \mathcal{N}(2,0.3)$-an asymmetric trimodal example with  one mode well separated from the other two modes. In Figure \ref{fig:simulated2} top row we illustrate the median, best and worst performance out of $100$ samples of size $100$ for the two examples. In Table \ref{tablede2} we present the quantitative performance analysis. 

One important observation is that the proposed method has much higher computation cost compared to the competitors because of the {\it GlobalSearch} toolbox used. For the symmetric unimodal example, the numerical performance with or without using the {\it GlobalSearch} toolbox is very similar, and hence the performance is presented without using the {\it GlobalSearch} toolbox to illustrate the difference in computation cost. For all other examples, {\it GlobalSearch} toolbox is used.
 \begin{figure}[t!]
\begin{center}
\begin{tabular}{ccc}
\includegraphics[width=0.33\linewidth]{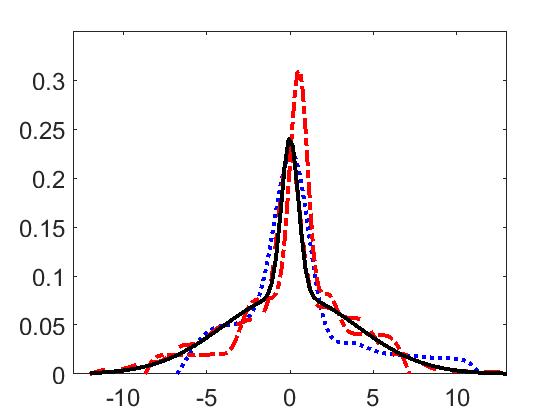} &
\includegraphics[width=0.33\linewidth]{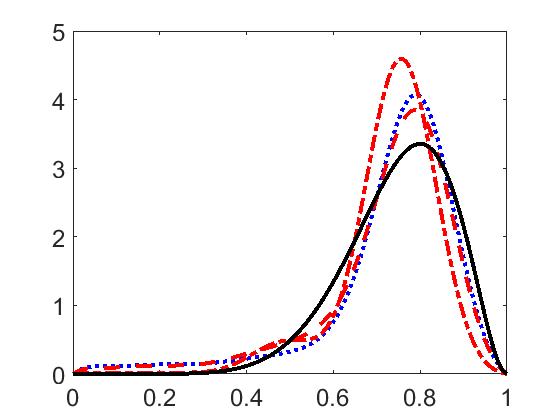}&
\includegraphics[width=0.33\linewidth]{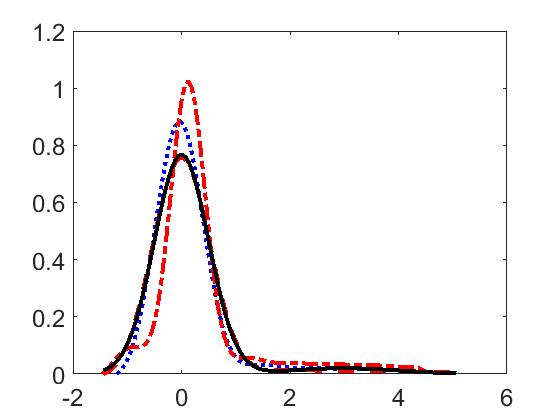}\\
\includegraphics[width=0.33\linewidth]{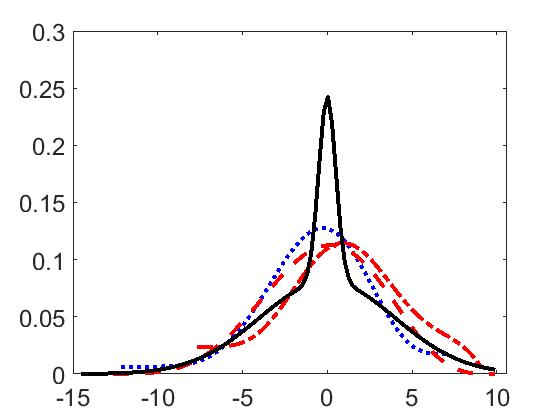} &
\includegraphics[width=0.33\linewidth]{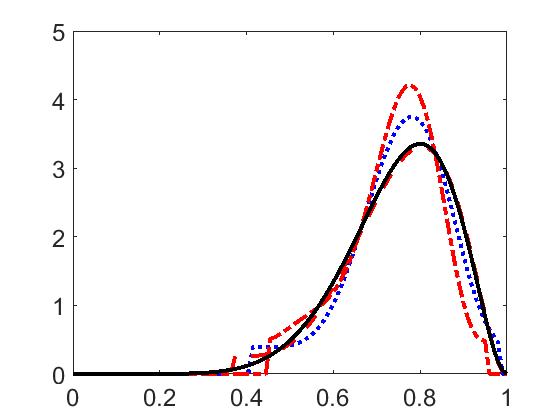}&
\includegraphics[width=0.33\linewidth]{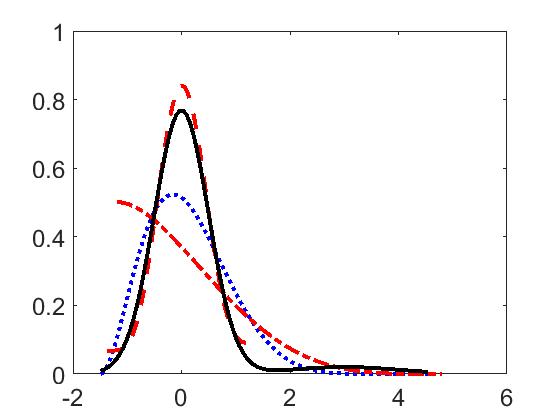}\\
\end{tabular}
\caption{\it The figure illustrates the true density(solid line) ;the estimated density with best performance(dashed line); the median performance(dotted line) and the worst performance(dashed-dotted line) according to $L_2$ norm. The panels correspond to the four simulated examples in order from top left to bottom right. }
\label{fig:simulated}
\end{center}
\end{figure}

\begin{table}[htbp] 
\caption{\it A comparison of the performances of \texttt{umd} package and Warped estimate for simulated unimodal examples.}
\label{tablede1}\par
\vskip .2cm
\centering
\begin{tabular}{|r|r|c|rrr|rrr|} \hline 
\multirow{2}{*}{Example} & \multicolumn{2}{c}{Method:} & \multicolumn{3}{c|}{Warped Estimate} & \multicolumn{3}{c|}{\texttt{umd} Estimate}\\ \cline{2-9}
& \multicolumn{1}{c}{$n$} &  \multicolumn{1}{c}{Norm}
           & \multicolumn{1}{c}{Mean} & \multicolumn{1}{c}{std.dev} & \multicolumn{1}{c|}{Time} & \multicolumn{1}{c}{Mean} & \multicolumn{1}{c}{std.dev} & \multicolumn{1}{c|}{Time} \\
 \hline
\multirow{9}{*}{Symmetric Unimodal} & \multirow{ 3}{*}{100} & $\lone$ &  1.1933 & 0.3038 &    & 1.5753 & 0.2202 & \\
& &  $\ltwo$ & 0.1898 & 0.0568 & $11$ sec & 0.2791 & 0.0138 & $1$ sec \\
& &  $\linf$ &  0.0755 & 0.0299 &  & 0.1243 & 0.0099 & \\ \cline{2-9} 
& \multirow{ 3}{*}{500} & $\lone$ &  0.5746 & 0.1131 &    & 1.1948 & 0.1109 & \\
& &  $\ltwo$ & 0.0953 & 0.0248 & $23$ sec & 0.2289 & 0.0050 & $1$ sec \\
& &  $\linf$ &  0.0409 & 0.0149 &  & 0.1109 & 0.0063 & \\ \cline{2-9}
 & \multirow{ 3}{*}{1000} & $\lone$ &  0.4786 & 0.2905 &    & 1.1325 & 0.0629 &  \\
& &  $\ltwo$ &  0.0834 & 0.0642 & $31$ sec & 0.2238 & 0.0036 & $1$ sec \\
& &  $\linf$ & 0.0371 & 0.3376 &  & 0.1117 & 0.0052 &  \\ \hline
\multirow{9}{*}{Skewed Unimodal} & \multirow{ 3}{*}{100} & $\lone$ &  19.4054 & 4.3991 &   & 14.0244 & 4.7563 &  \\
& &  $\ltwo$ &  2.9589 & 0.7715 & $305$ sec & 2.1081 & 0.7414 & $1$ sec \\
& &  $\linf$ &  0.8517 & 0.2914 &   & 0.5668 & 0.2074 & \\ \cline{2-9}
& \multirow{ 3}{*}{500} & $\lone$ &  12.9066 & 3.0470 &    & 7.6131 & 2.3679 & \\
& &  $\ltwo$ & 1.9930 & 0.5267 & $259$ sec & 1.1735 & 0.3838 & $1$ sec \\
& &  $\linf$ &  0.5866 & 0.1832 &  & 0.3294 & 0.1141 & \\ \cline{2-9}
 & \multirow{ 3}{*}{1000} & $\lone$ &  12.0474 & 2.5418 & & 5.6584 & 1.4165 &   \\
& &  $\ltwo$ & 1.8592 & 0.4427 & $341$ sec & 0.8779 & 0.2370 & $1$ sec\\
& &  $\linf$ &  0.5485 & 0.1765 &  & 0.2485 & 0.0732 &
\\ \hline
\multirow{9}{*}{Contaminated Unimodal} & \multirow{ 3}{*}{100} & $\lone$ &  3.0600 & 1.5574 &  & 6.6567 & 1.4372 & \\
& &  $\ltwo$ &  0.4385 & 0.2258 & $277$ sec & 0.9532 & 0.2374 & $1$ sec   \\
& &  $\linf$ &  0.1136 & 0.0628 &  & 0.1455 & 0.0538 &  \\ \cline{2-9}
& \multirow{ 3}{*}{500} & $\lone$ &  1.2348 & 0.5206 &    & 3.4151 & 0.8655 & \\
& &  $\ltwo$ & 0.1893 & 0.0879 & $301$ sec & 0.5106 & 0.1568 & $1$ sec \\
& &  $\linf$ &  0.0510 & 0.0268 &  & 0.1455 & 0.0538 & \\ \cline{2-9}
 & \multirow{ 3}{*}{1000} & $\lone$ &  0.8319 & 0.3172 &   & 3.1453 & 0.8934 &  \\
& &  $\ltwo$ & 0.1247 & 0.0563 & $301$ sec & 0.4616 & 0.0879 & $1$ sec\\
& &  $\linf$ &  0.0363 & 0.1277 &  & 0.0502 & 0.0538 &
\\ \hline
\end{tabular}
\end{table}

\begin{table}[htbp] 
\caption{\it A quantitative analysis of the performance of Warped Estimate for simulated bimodal and trimodal dataset.}
\label{tablede2}\par
\vskip .2cm
\centering
\begin{tabular}{|r|c|rrr|rrr|} \hline 
  \multicolumn{2}{|c|}{Example:} & \multicolumn{3}{c|}{Bimodal density} & \multicolumn{3}{c|}{Trimodal density}\\ \cline{1-8}
 \multicolumn{1}{|c}{$n$} &  \multicolumn{1}{c|}{Norm}
           & \multicolumn{1}{c}{Mean} & \multicolumn{1}{c}{std.dev} & \multicolumn{1}{c|}{Time} & \multicolumn{1}{c}{Mean} & \multicolumn{1}{c}{std.dev} & \multicolumn{1}{c|}{Time} \\
 \hline
 \multirow{ 3}{*}{100} & $\lone$ &  4.3429 & 1.2332 &    & 6.7299 & 1.6367 & \\
 &  $\ltwo$ & 0.6575  & 0.2049 & $125$ sec & 0.9075 & 0.2344 & $105$ sec \\
 &  $\linf$ &  0.2089 & 0.0850 &  & 0.2419 & 0.0867 & \\ \cline{1-8} 
 \multirow{ 3}{*}{500} & $\lone$ &  2.4727 & 0.5755 &    & 3.4841 & 1.2778 & \\
 &  $\ltwo$ & 0.3839 & 0.1103 & $143$ sec & 0.4816 & 0.1737 & $131$ sec \\
 &  $\linf$ &  0.1337 & 0.0502 &  & 0.1351 & 0.0538 & \\ \cline{1-8}
  \multirow{ 3}{*}{1000} & $\lone$ & 1.9942 & 0.5042 &    &  3.0489 & 1.7033 &  \\
 &  $\ltwo$ &  0.3100 & 0.0999 & $185$ sec & 0.4330 & 0.2353 & $311$ sec \\
 &  $\linf$ & 0.1095 & 0.0444 &  &  0.1246 & 0.0648 &  \\ \hline
\end{tabular}
\end{table}

\begin{table}[htbp] 
\caption{\it A quantitative analysis of the performance of Warped Estimate for simulated bimodal and trimodal dataset.}
\label{tablede3}\par
\vskip .2cm
\centering
\begin{tabular}{|r|c|rrr|rrr|} \hline 
  \multicolumn{2}{|c|}{$p_0\sim \mathcal{N}(0,1)I_{[0,1]}$} & \multicolumn{3}{c|}{EpisplineDensity estimate} & \multicolumn{3}{c|}{Warped Estimate}\\ \cline{1-8}
 \multicolumn{1}{|c}{$n$} &  \multicolumn{1}{c|}{Norm}
           & \multicolumn{1}{c}{Mean} & \multicolumn{1}{c}{std.dev} & \multicolumn{1}{c|}{Time} & \multicolumn{1}{c}{Mean} & \multicolumn{1}{c}{std.dev} & \multicolumn{1}{c|}{Time} \\
 \hline
 \multirow{ 3}{*}{500} & $\lone$ &  8.7334 & 2.1415 &    & 5.9202 & 2.8516 & \\
 &  $\ltwo$ & 1.3269  & 0.5033 & $6$ sec & 0.7079 & 0.3346 & $180$ sec \\
 &  $\linf$ &  0.6167 & 0.4461 &  & 0.1538 & 0.0758 & \\ \cline{1-8} 
\end{tabular}
\end{table}

 \begin{figure}[t!]
\begin{center}
\begin{tabular}{cc}
\includegraphics[width=0.5\linewidth]{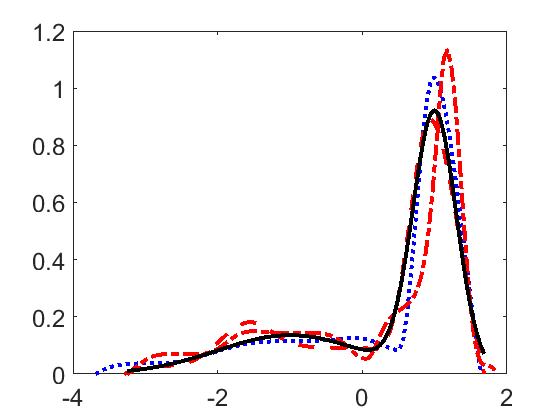} &
\includegraphics[width=0.5\linewidth]{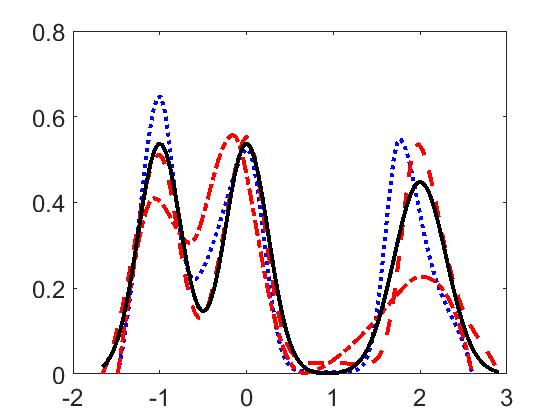}\\
\includegraphics[width=0.5\linewidth]{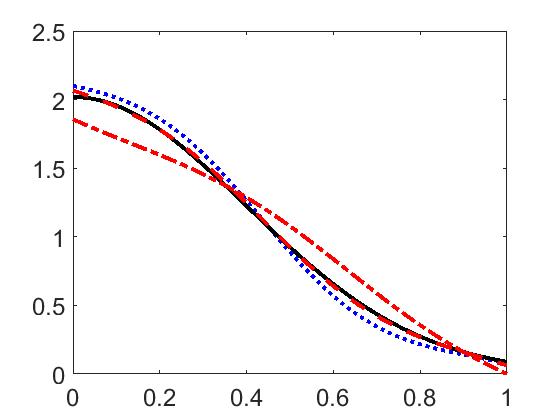} &
\includegraphics[width=0.5\linewidth]{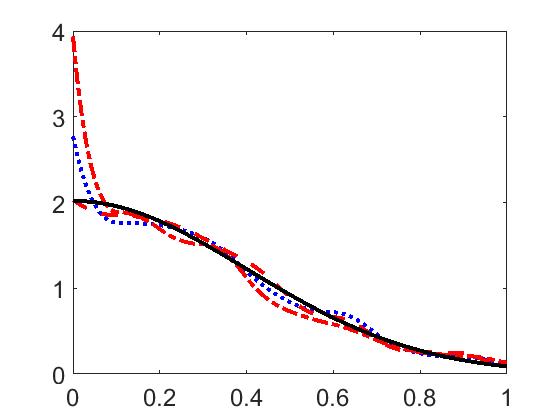} \\
\end{tabular}
\caption{\it The figure illustrates the true density(solid line) ;the estimated density with best performance(dashed line); the median performance(dotted line) and the worst performance(dashed-dotted line) according to $L_2$ norm. The panels correspond to the four simulated examples in order from top left to bottom right. }
\label{fig:simulated2}
\end{center}
\end{figure}

\section{Extension to more general shapes}
\label{sec:extensiontomoregeneralshapes}

Upto this point we have restricted ourselves to density estimates which are zero at the boundary even though the true density might not be exactly zero. Also the estimation has inherently assumed that the $M$ modes lie in the interior of the support and not on the boundary. As indicated in the simulation studies, the method has very good numerical performance for densities which decay at the boundaries. However, the proposed framework allows a easy extention to densities which may have (1) modes located at the boundaries, (2) compact support with significantly large value at the boundaries, by simply considering the height ratios at the boundaries as extra parameters. Essentially this extension requires knowledge of the exact sequence of the modes and antimodes in order to construct the correct function template $g_{\lambda}$ and the correct constraints for the parameters $\lambda_i$. (Note that previously we had indexed the template by $g_{\lambda}^{\omega}$ and we fixed the boundary values of $g$ to be $\omega$). For example, for an $N$-shaped density, we need the knowledge that the function is initially increasing,then decreasing and finally increasing, and hence we can create an $N$ shaped template. Once the template is constructed the rest of the procedure remains the same. Another special example are monotone densities, where the mode is at one of the boundaries. In such a scenario, one can construct the template by setting the modal value of $g$ to be $1$ and estimate the other boundary value $\lambda_1$ with appropriate constraint. The bottom row of Figure \ref{fig:simulated2} considers an example of a monotonically decreasing density, a $\mathcal{N}(0,0.4)$ truncated to $[0,1]$. As a comparison we have used the \texttt{episplineDensity} package and have considered $100$ samples of sample size $500$. The bottom left panel of Figure \ref{fig:simulated2} shows the best, median and worst performance out of the $100$ samples for the warped estimate. The right panel shows the same for the \texttt{episplinedensity} estimate. The performance of the warped estimate is better overall, and especially at the left boundary. Table \ref{tablede3} presents the quantitative comparison of the performances.

Finally suppose a density has a flat spot at a modal (or antimodal) location. This indicates that the modes are not well defined but is actually an interval. The framework theoretically accomodates such an information by simply adding a flat spot in the template function at the desired location. 
Thus, we can extend the idea of ``shape" of a continuous density function to be identified as an ordered sequence of increasing, decreasing or flat pieces that form the entire density function. For example, a simple bimodal density function can be identified by the sequence {\it increasing-decreasing-increasing-decreasing}. A function with a unique modal interval can be described as {\it increasing-flat-decreasing}. If this sequence is known, then simply constructing a template with the same sequence allows us to provide a maximum likelihood density estimate within the class of densities satisfying that shape sequence.

 \begin{figure}[htbp]
\begin{center}
\begin{tabular}{cc}
\includegraphics[width=0.5\linewidth]{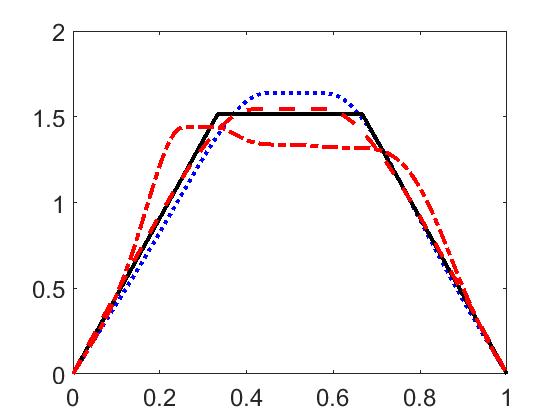} &
\includegraphics[width=0.5\linewidth]{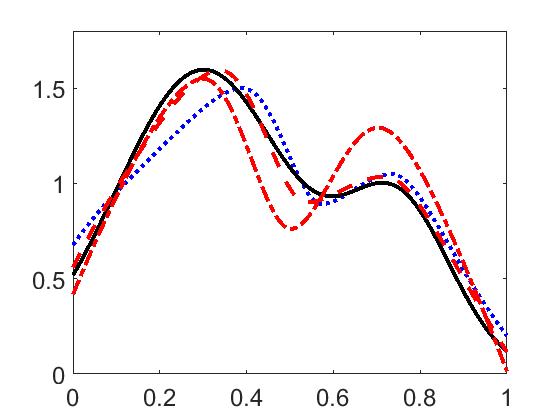}\\
\end{tabular}
\caption{\it The figure illustrates the true density(solid line) and the estimated density with median performance(dotted line); the best performance (dashed line) and the worst performance(dashed dotted line) for two examples with estimates unconstrained at the boundaries.}
\label{sde}
\end{center}
\end{figure}

In practice, we have used MATLAB function {\it fmincon} for optimization purposes. However, estimating the correct height at the boundaries takes a large sample size using the {\it fmincon} implementation to achieve a satisfactory and stable performance. Figure \ref{sde}  shows two examples, 
\begin{enumerate}
\item
$p_0 \propto xI_{x \in [0,1/3]} + 1/3I_{x \in [1/3,2/3]} +(1-x)I_{x \in [2/3,1]}$ and zero otherwise - A  density function with a flat modal region.

\item
$p_0 \propto 3/4\mathcal{N}(0.3,{0.2}^2)I_{[0,1]} + 1/4\mathcal{N}(0.75,{2}^2)I_{[0,1]}$ - A bimodal density function truncated to $[0,1]$.
\end{enumerate}

The left panel of Figure \ref{sde} shows the best, median and worst performance out of $100$ samples of size $500$ from the density with flat spot. The right panel shows the same from sample size $1000$ for the truncated bimodal density.

\section{Extension to conditional density estimation}
\label{sec:extensiontoconditional}

The proposed framework for modality constrained density estimation extends naturally to modality constrained conditional density estimation setups.  
Consider the following setup: Let $X$ be a fixed one-dimensional random variable with a positive density on its support. 
Let $Y \sim  f_{tX}(m(X),\sigma_X^2)$, where $f_{tX}$ is the unknown conditional density that changes smoothly with $X$;
$m(X)$ is the unknown mean function, assumed to be smooth; and, $\sigma_X^2$ is the unknown variance, 
which may or may not depend on $X$. Conditioned on $X$, $Y$ is assumed to have a univariate, continuous distribution with  support on 
interval $[A,B]$, has a known $M$ modes in the interior of $[A,B]$, and $f_{tX}(A)=f_{tX}(B)=0$.
 We observe the pairs 
$(Y_i,X_i),i=1,\dots,n$, and are interested in recovering the conditional density $f_{tX}(m(X),\sigma^2)$ at a particular location of $X$, henceforth referred to as $x_0$.
The estimation is again initialized with an $M$ modal template function $g_{\lambda}^{\omega}$. However, since $f_{tX}$ varies smoothly with $X$, we assign more importance to observations closer to the location $x_0$ than observations that are further away, and hence, we perform weighted maximum likelihood function to estimate the necessary parameters.
\begin{equation}
(\hat{c}_{x_0},\hat{\lambda}_{x_0}) =  \argmax_{c \in V_{\pi}^J, \lambda \in \Lambda} \left(  \sum_{i=1}^n \Bigg[ \log \bigg(g_{\lambda}^{\omega} \left(\gamma_c(x_i)\right)/\int_0^1 \left(g_{\lambda}^{\omega} \left(\gamma_c(t)\right)dt\right)\bigg) \Bigg] W_{x_0,i}\right),\ \ 
\gamma_c = G(c)\ .
\label{eq:cdeopt}
\end{equation}

where $W_{x_0,i}$ is the localized weight associated with the $i$th observation, calculated according to:
$$
W_{x_0,i}=\frac{{\cal N}({\|X_i-x_0\|}_2/h(x_0);0,1)}{\sum_{j=1}^{n}{\cal N}({\|X_j-x_0\|}_2/h(x_0);0,1)}
$$
where ${\cal N}(\cdot,0,1)$ is the standard normal pdf and $h(x_0)$ is the parameter that controls the relative weights associated with the observations.  However, weights defined in this way results in higher bias because information is being borrowed from all observations. As discussed in an example in \citet{bashtannyk2001bandwidth}, we allow only a specified fraction of the observations $X_i$ to have a positive weight. However, using too small a fraction will result in unstable estimates and poor practical performance because the effective sample size will be too small. Hence we advocate using the nearest $50\%$ of the observations (nearest to the target location) for borrowing information and then calculating the weights for this smaller sample as defined before.
The parameter $h(x_0)$ is akin to the bandwidth parameter associated with traditional kernel methods for density estimation, for the predictors $X$. 
A very large value of $h(x_0)$ distributes approximately equal weight to all the observations, whereas a
very small value considers only the observations in a neighborhood around $x_0$. 
The parameter $h(x_0)$ can be chosen via any standard cross validation based bandwidth selection method, for practical purposes. For our purposes we use an adaptive bandwidth selection method to save computation time, when the predictors are independent of each other:

The parameter $h(x_0)$ is chosen according 
to the location $x_0$ using a two-step procedure:
\begin{enumerate}
\item
Compute a standard kernel density estimate $\hat{K}$ of the predictor space using a fixed bandwidth chosen according to any standard criterion. 
For our purposes, we simply used the \texttt{ksdensity} estimate inbuilt in MATLAB which chooses the bandwidth optimal for normal densities. Let $h$ be the fixed bandwidth used.

\item
Then, set the bandwidth parameter $h( x_0)$ at location $x_0$ to be $h(x_0)=h/\sqrt{\hat{K}(x_0)}$.
 
\end{enumerate}

The intuition is that $h$ controls the overall smoothing of the predictor space based on the sample points, and the $\sqrt{\hat{K}(x_0)}$ stretches or shrinks the bandwidth at the particular location. At a sparse region, increased borrowing of information from the other data points is desirable in order to reduce the variance of the estimate, whereas in dense regions a
reduced borrowing of information from far away points reduces the bias of the density estimates. A location from a sparse region is expected to have a low density estimate, and a location from a dense region is expected to have a high density estimate. Hence, varying the bandwidth parameter inversely with the density estimate helps adapt to the sparsity around the point of interest. The choice of the adaptive bandwidth parameter is motivated from the variable bandwidth kernel density estimators discussed in \citet{terrell1992variable}, \citet{van2003adaptive} and \citet{abramson1982bandwidth}, among others. 
%
%

As illustrative examples we consider two setups: (1)$X\sim \mathcal{N}(0,1),Y|X \sim DExp( {(2X-1)}^2,1)$, a unimodal conditional density and (2) $X\sim \mathcal{N}(0,1),Y|X\sim  0.5\mathcal{N}(X-1.5,{0.5}^2) + 0.5 \mathcal{N}(X+1.5,{0.5}^2)$, a bimodal conditional density. In both cases we study $100$ samples of size $100$ and $1000$ and compute the conditional density at the $25$th, $50$th and $75$th quantile of the predictor support.  Figure \ref{fig:simulatedcde} illustrates the best, worst and the median performance among the $100$ samples in each scenario. For sample size $100$ (first and third row), the performance is slightly unstable and the worst performances often has a bias and is wiggly in nature. Naturally for larger sample size $1000$ (second and fourth row), the results are much more stable. Also noteworthy is the more pronounced bias for the conditional densities evaluated at the $25$th and $75$th quantiles because of borrowing of information via weighted likelihood estimation. However, the bias is almost absent for sample size $1000$. The quantitative performance based on average loss functions is presented in Table \ref{tablecde1}.

 \begin{figure}[t!]
\begin{center}
\begin{tabular}{ccc}
\includegraphics[width=0.33\linewidth]{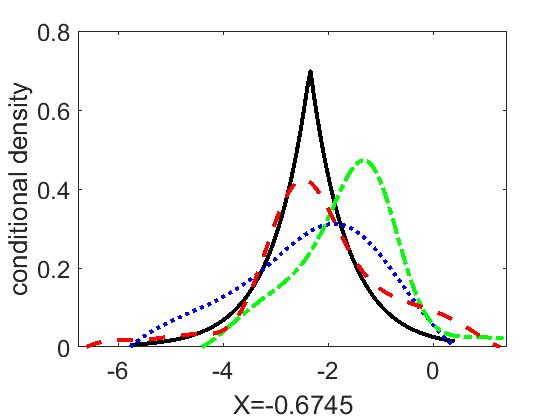} &
\includegraphics[width=0.33\linewidth]{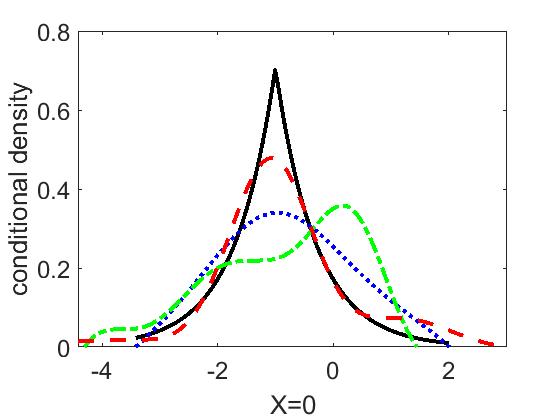}&
\includegraphics[width=0.33\linewidth]{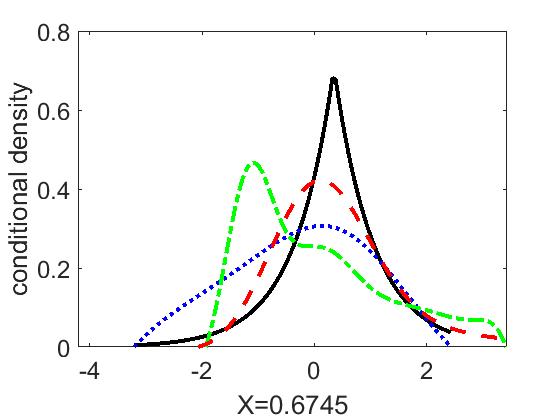}\\
\includegraphics[width=0.33\linewidth]{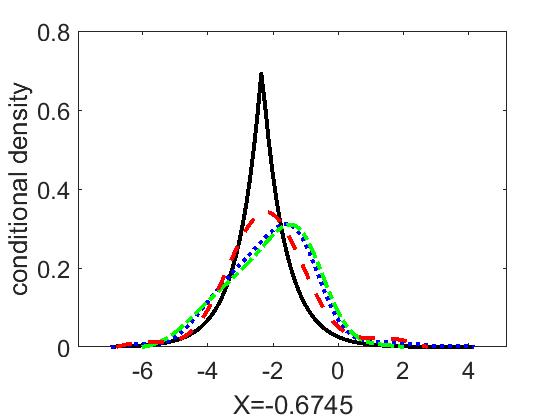} &
\includegraphics[width=0.33\linewidth]{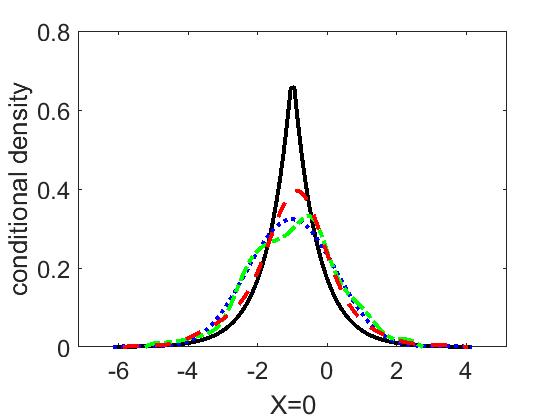}&
\includegraphics[width=0.33\linewidth]{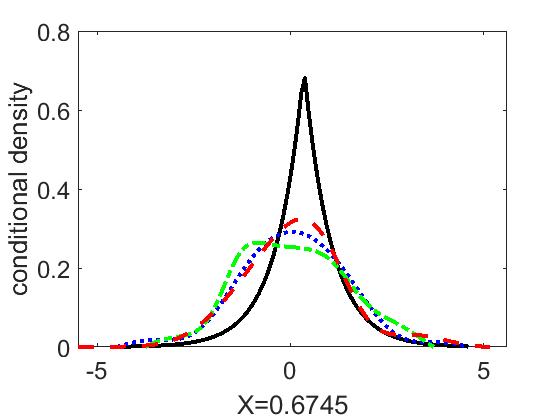}\\
\includegraphics[width=0.33\linewidth]{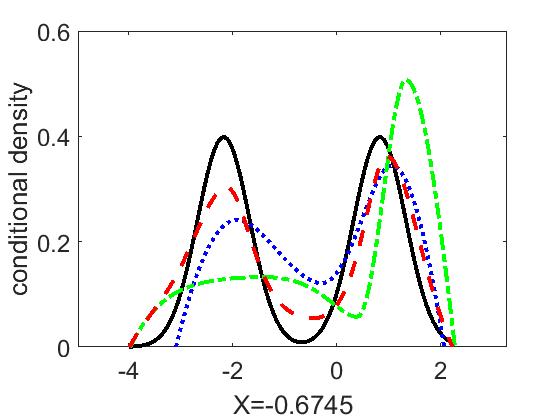} &
\includegraphics[width=0.33\linewidth]{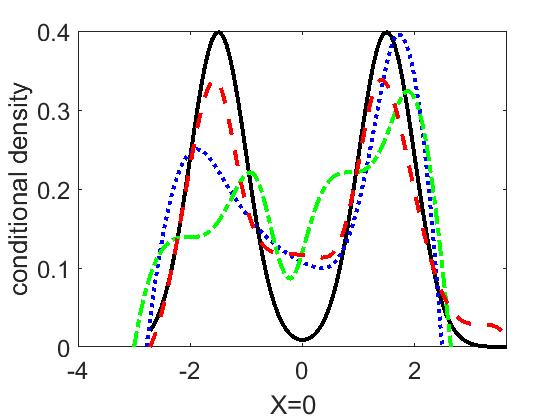}&
\includegraphics[width=0.33\linewidth]{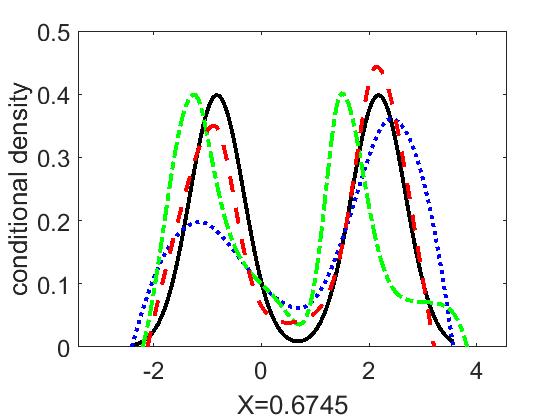}\\
\includegraphics[width=0.33\linewidth]{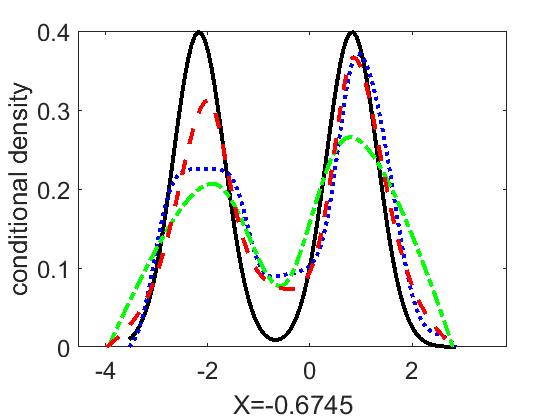} &
\includegraphics[width=0.33\linewidth]{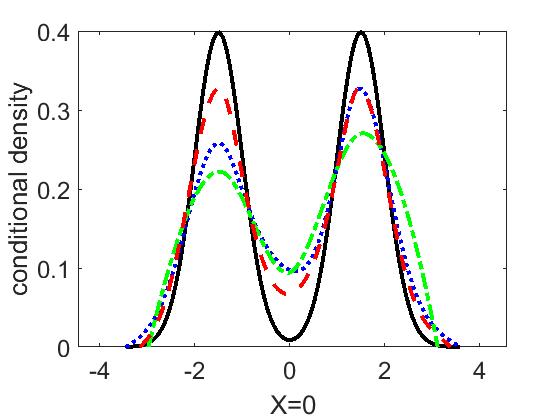}&
\includegraphics[width=0.33\linewidth]{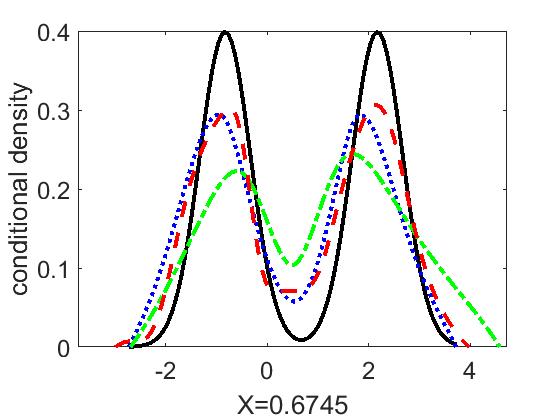}\\
\end{tabular}

\caption{\it The figure illustrates the true density(solid line) and the estimated density with median performance(dotted line); the best performance (dashed line) and the worst performance(dashed dotted line) at three different locations in the support of the predictors.}
\label{fig:simulatedcde}
\end{center}
\end{figure}

\begin{table}[t!] 
\caption{\it A quantitative evaluation of the performance of Warped estimate for two simulated conditional density examples.}
\label{tablecde1}\par
\vskip .2cm
\centering
\begin{tabular}{|r|r|c|rr|rr|rr|} \hline 
\multirow{2}{*}{Example} & \multicolumn{2}{c}{Location:} & \multicolumn{2}{c|}{$25$th quantile} & \multicolumn{2}{c|}{$50$th quantile} & \multicolumn{2}{c|}{$75$th quantile}\\ \cline{2-9}
& \multicolumn{1}{c}{$n$} &  \multicolumn{1}{c}{Norm}
           & \multicolumn{1}{c}{Mean} & \multicolumn{1}{c|}{std.dev} & \multicolumn{1}{c}{Mean} & \multicolumn{1}{c|}{std.dev} & \multicolumn{1}{c}{Mean} & \multicolumn{1}{c|}{std.dev}\\
 \hline
\multirow{6}{*}{Unimodal cde} & \multirow{ 3}{*}{100} & $\lone$ &  7.9623 & 2.0550 & 6.9829 & 1.9716 &  8.2243 & 2.4475   \\
& &  $\ltwo$ & 1.1658 & 0.2539 & 1.0132 & 0.2388 & 1.1884 & 0.2876 \\
& &  $\linf$ &  0.4056 & 0.0570 & 0.3586 & 0.0589 & 0.4094 & 0.0595 \\ \cline{2-9} 
 & \multirow{ 3}{*}{1000} & $\lone$ &  5.1280 & 0.7392 &  4.1239 & 0.6308 & 5.2136 & 0.7194   \\
& &  $\ltwo$ &  0.9271 & 0.0929 & 0.7537 & 0.0812 & 0.9297 & 0.0846\\
& &  $\linf$ & 0.3977 & 0.0275 & 0.3494 & 0.0256 & 0.3966 & 0.0239 \\ \hline
\multirow{6}{*}{Bimodal cde} & \multirow{ 3}{*}{100} & $\lone$ &  8.3386 & 1.5436 & 7.0026 & 1.2024 & 7.8851 & 1.4847  \\
& &  $\ltwo$ &  0.9983 & 0.1802 & 0.8374 & 0.1307 & 0.9478 & 0.1695 \\
& &  $\linf$ &  0.2044 & 0.0384 & 0.1773 & 0.0349 & 0.2015 & 0.0430  \\ \cline{2-9}
 & \multirow{ 3}{*}{1000} & $\lone$ &  5.8890 & 0.6466 & 4.9654 & 0.9002 & 5.9918 & 0.6902\\
& &  $\ltwo$ & 0.7201 & 0.0756 & 0.6111 & 0.1001 & 0.7285 & 0.0766\\
& &  $\linf$ &  0.1574 & 0.0205 &  0.1406 & 0.0255 & 0.1561 & 0.0180
\\ \hline
\end{tabular}
\end{table}

\section{Application to speedflow data}
\label{sec:application}

As an application of modality constrained conditional density estimation, we use the speed flow data for Californian driveways from the package \texttt{hdrcde} in R. The scatterplot shown in Figure \ref{fig:real} shows the distinct bimodal nature of the speed distribution for traffic flow between $1000$ and $1620$ vehicles per lane per hour, corresponding to uncongested and congested traffic. This range of traffic flow where a bimodal nature is apparent is already studied in \citet{einbeck2006modelling}. They study that beyond traffic flow of $1620$ the regression curves corresponding to uncongested and congested traffic are no longer distinguishable. So, we consider the speed flow in that range ($772$ observations) and compute the conditional density of the speed with bimodality constraint on the shape, given flow$=1400$ using our prescribed $50 \%$ of the $772$ observations. The middle panel of Figure \ref{fig:real} (solid line) shows the conditional density estimate for flow$=1400$ using the proposed approach. The left mode is $35.56$ mph and the right mode is $59.01$. \citet{einbeck2006modelling} also obtains a very similar conditional density estimate. The left mode in their case is at $32.65$ mph and the right mode is at $59.18$. On the other hand if we carry out a traditional conditional  density estimation using \texttt{NP} package, we see several spurious bumps as shown in the middle panel of Figure \ref{fig:real} (dotted line) and zoomed in on the right panel. The bumpy nature is present in the \texttt{NP} estimate constructed using $772$ observations (not presented) as well as only using $50\%$ of the observations as in our approach. This results in over-interpreting the tail and consequently a lack of interpretability for the modes themselves. Thus constraining the number of modes clearly helps lending interpretability to the resultant density shape. 


 \begin{figure}[t!]
\begin{center}
\begin{tabular}{ccc}
\includegraphics[width=0.33\linewidth]{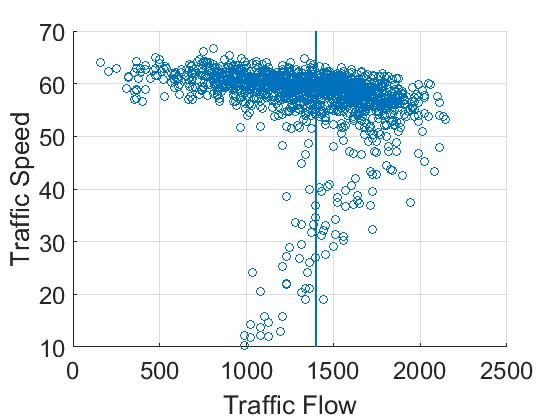} &
\includegraphics[width=0.33\linewidth]{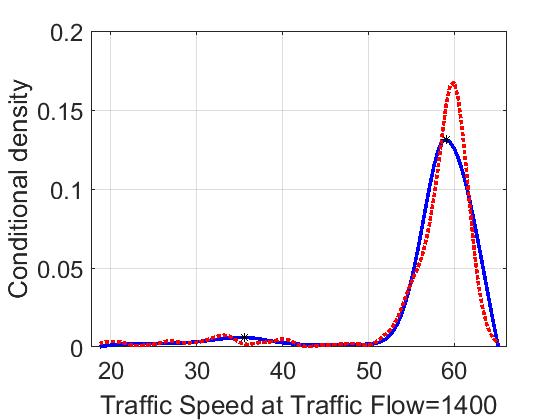}&
\includegraphics[width=0.33\linewidth]{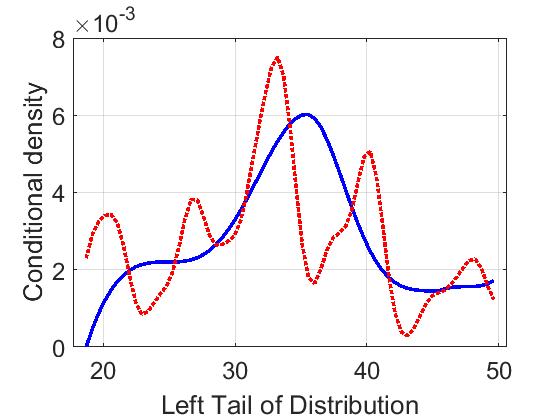}\\
\end{tabular}

\caption{\it The figure illustrates the scatterplot(top left) and the estimated density of traffic speed at traffic flow $1400$(top right)for warped method (solid) and NP package(dotted).}
\label{fig:real}
\end{center}
\end{figure}

\section{Discussion}
\label{sec:discussion}

Density estimation and shape constrained density estimation are very rich topics of research in Statistics. The current paper focuses on introducing a novel framework using geometric tools which enables one to perform shape constrained density estimation with a broader notion of shapes than before. Specifically, exploiting the geometry of the group of diffeomorphisms, one can shift the problem of finding a density with the appropriate shape constraints to finding an appropriate diffeomorphism given an initial shape, based on available data. In recent years, most datasets on a variable of interest have associated covariates which make the problem of conditional density estimation very useful and practically relevant. An advantage of the proposed framework is the easy extendibility to the conditional density estimation problem via a weighted maximum likelihood objective function.  Theoretically, the framework introduced is the first that can perform any $M$ modality constrained density estimation. However practically the performance suffers when the constrained shape is too complicated or if the number of modes $M$ is very high (greater than $4$) because the inbuilt solver {\it fmincon} gets stuck in local suboptimal solutions resulting in unsatisfactory density estimates.

Since the paper primarily focuses on introducing the framework and the group action that enables shape constrained estimation, it has only lightly touched upon or not explored many associated problems of density estimation. For example, the choice of the number of basis elements for tangent space representation, the choice of the basis set itself, or the choice of penalty for penalized estimation and boundary estimation are very rich and important problems themselves in their own right. This paper simply uses AIC as the penalty to select the number of basis elements because in comparison, BIC tends to choose insufficient number of parameters. Also, experiments using a Meyer basis set for the tangent space representation of the diffeomorphisms yielded similar results, though the Meyer wavelets seemed to require more observations than Fourier basis set to obtain satisfactory results. Keeping in mind that the basis set representation is for approximating the warping functions and not the density functions directly, one can choose different basis sets for a comparative study of performances. The paper follows \cite{Turnbull2014-ag} for choosing the boundaries. 

For conditional density estimation, the weights defined as gaussian kernel can also be defined using any other kernel. The choice of gaussian kernel (and the \ltwo  loss function) was as an illustration. A possible extension not explored in the paper is to develop the framework in situations where multiple or very high number of covariates are present. Currently the bandwidth parameter is chosen adaptively based on a kernel density estimate at the location of the (scalar) covariate. It can be directly extended to $d$ covariate scenario using a $d$ variate kernel density estimate at the location of the predictors. However, such an estimate suffers from the curse of dimensionality. In applications where only a few of the covariates are relevant to the response variable, \citet{wasserman2006rodeo} developed a technique to identify the relevant variables and also obtain the corresponding bandwidth parameters. Using the obtained bandwidth parameters, one can redefine the weights and perform weighted likelihood maximization to produce a conditional density estimate.

\appendix

\section{Proof of Theorem \ref{thm2}} \label{app:thm2}

First we set some notations and some preliminary definitions. $M$ is always used to represent the number of modes. Let $g_{\lambda}^{\omega}$ denote the $M$-modal template defined earlier as a function of $\lambda$. Here ${\mathbf \lambda}$ denotes the vector $(\lambda_1 \cdots \lambda_{2M-2})$, corresponding to the $2M-2$ height ratios of the last $2M-2$ critical points with respect to the first critical point. Let $k_n$ be the number of basis elements used for approximating the warping function $\gamma$. Let $c= (c_1, \cdots , c_{k_n})$ be the corresponding coefficient vector. Now, define $\theta_n= (c_1, \cdots c_n, \lambda_1 , \cdots, \lambda_{2M-2})$. In what follows, $c$ is used to represent the coefficient vectors. $B_i$ denotes the $i$th basis element for the tangent space representation of warping functions. $\gamma_c$ is used to represent the warping function corresponding to the coefficient vector $c$. $l_1, l_2 ,\cdots, C, C_1, \cdots $ represent specific constants. $M_0, M_1, M_2, \cdots$ represent generic constants that can change values from step to step but are otherwise independent of other terms.

 
Let $\lambda_0 \in \real^{2M-2}$ be the {\it height ratio vector} for $p_0$, as defined in Section 2. Then from Theorem \ref{gamexist} there exists an infinite dimensional ${c_0}$ such that $p_0$ can be represented as 
\[
p_0 = (g_{\lambda_0}^0 \circ \gamma_{c_0}) / \int_0^1 (g_{\lambda_0}^0 \circ \gamma_{c_0}) dt.
\]   
Note that for each $t \in [0,1]$, $\|\sum_{i=1}^{\infty} c_iB_i(t)\|=\sqrt{\int{( \sum_{i=1}^{\infty} c_iB_i(t))}^2}<2\pi$. This corresponds to $ \max_{1\leq j \leq k_n} \abs{c_{0j}}<l_0$ and thus $|c_{0i}|<l_0$ for all $i$, for some $l_0$. Then the parameter space for $\mathcal{F}$ is $\Theta=\{(c,\lambda): c \in {[-l_0,l_0]}^{\infty}, \lambda \in \Lambda \subset {(0,\infty)}^{2M-2}\}$. Let $\omega =\omega (n) =  \Omega /\log n$.where $\Omega$ is a constant.  Let $r_n^u=\Omega_1\log n$
and $r_n^l=\Omega_1/\log n$ where $\Omega_1 <\Omega$ is some constant. Define $\mathcal{F}_n$ as the approximating space of densities for $\mathcal{F}$. Define $\Theta_n=\{\theta_n=(c,\lambda) |c \in {[-l_0,l_0]}^{k_n}, \lambda \in {(r_n^l,r_n^u]}^{2M-2}\} $ as the parameter space for the approximating space $\mathcal{F}_n$. Then $\mathcal{F}_n= (g_{\lambda}^{\omega} \circ \gamma_c) / \int_0^1 (g_{\lambda}^{\omega} \circ \gamma_c) dt$ where $\theta_n=(c,\lambda) \in \Theta_n$. 
We use the method of sieve maximum likelihood estimation to obtain the estimate in the approximating space $\mathcal{F}_n$ of $\mathcal{F}$ and to derive an upper bound of the convergence rate of the density estimate to the final density.

We call a finite set $\{(f_{j}^{L},f_{j}^{U}),j=1,\dots,N\}$ a Hellinger $u$-bracketing of $\mathscr{F}_n$ if 
${\|{f_{j}^{L}}^{1/2}-{f_{j}^{U}}^{1/2}\|}_{2} \leq u$ for $j=1,\dots,N$, and for any $p \in \mathscr{F}_n$, there is a $j$ such that $f_{j}^{L} \leq p \leq f_{j}^{U}$. 
Let $H(u,\mathscr{F}_n)$ denote the Hellinger metric entropy of $\mathscr{F}_n$, defined as the logarithm of the cardinality of the $u$-bracketing of $\mathscr{F}_n$ of the smallest size.  
To control the approximation error of $\mathcal{F}_n$ to $\mathcal{F}$, \citet{wong1995probability} introduces a family of discrepancies. They define $\delta_n (p_0,\mathscr{F}_n)=\text{inf}_{p \in \mathscr{F}_n} \rho (p_0,p)$, called the $\rho$-approximation error at $p_0$.
The control of the approximation error of $\mathscr{F}_n$ at $p_0$ is necessary for obtaining results on the convergence rate for sieve MLEs.
We follow \citet{wong1995probability} to introduce a family of indexes of discrepency in order to formulate the condition on the approximation error of $\mathscr{F}_n$.
Let
\[
Z_{\alpha} (x) = \left\{\begin{array}{lr}(1/\alpha)[x^{\alpha} -1], -1<\alpha<0 \text{ or }0<\alpha \leq 1\\
\log{x}, \text{  if }\alpha=0+. 
\end{array}
\right.
\]
Set $x=p_0/p$ and define $\rho_{\alpha} (p_0,p)=E_pZ_{\alpha}(X)=\int p_0Z_{\alpha}(p_0/p).$
We define $\delta_n (\alpha) =\inf _{p \in \mathscr{F}_n} \rho_{\alpha} (p_0,p)$. For our purposes we set $\alpha=1$.
Thus we have $\delta_n (1) = \underset{p \in \mathcal{F}_n}{\mbox{inf}} \int {(p_0-p)}^2/p$.

Let $f_1$ and $f_2$ be two densities in $\mathcal{F}_n$. Let $\theta_1=({c_1},\lambda_1)$ and $\theta_2=({c_2},\lambda_2)$ be the corresponding parameters. $g_1^{\omega}$ and $g_2^{\omega}$ be the corresponding templates. Let $M$ be the number of modes and $\gamma_1$ and $\gamma_2$ be the warping functions corresponding to the coefficients. Then we have

\begin{lem}
 $|f_1 - f_2|\leq M_0 \sum_{i=1}^{k_n +2M-2}{|\theta_{1i}- \theta_{2i}|}$
\label{equiv1}, for some constant $M_0> 0$. 
\end{lem}

\begin{proof}
First, following the steps of \citet{dasgupta2017geometric} we observe that $|\gamma_1 (t) -\gamma_2 (t)| < M_2 \sum_{i=1}^{k_n}| c_{1i} - c_{2i} | < M_1 \sum_{i=1}^{k_n +2M-2}{|\theta_{1i}- \theta_{2i}|}$ since the $c_i$'s are simply the first few coordinates of $\theta$.
Next, observe that $
|g_1^{\omega} \circ \gamma_1 - g_2^{\omega} \circ \gamma_2| \leq |g_1^{\omega} \circ \gamma_1 - g_1^{\omega} \circ \gamma_2| +
|g_1^{\omega} \circ \gamma_2 - g_2^{\omega} \circ \gamma_2|$.
By construction, $g_1^{\omega}$ is Lipschitz continuous, and hence 
$
|g_1^{\omega} \circ \gamma_1 - g_1^{\omega} \circ \gamma_2|\leq M_2|\gamma_1 -\gamma_2| \leq M_3 \sum_{i=1}^{k_n +2M-2}{|\theta_{1i}- \theta_{2i}|}$.
Now, we have $|g_1^{\omega} \circ \gamma_2 - g_2^{\omega} \circ \gamma_2| \leq \underset{1\leq i \leq (2M-2)}{\max}|\lambda_{1i} - \lambda_{2i}| \leq M_2 \sum_{i=1}^{k_n +2M-2}{|\theta_{1i}- \theta_{2i}|}$.
Thus, it follows that $|g_1^{\omega} \circ \gamma_1 - g_2^{\omega} \circ \gamma_2| \leq M_1  \sum_{i=1}^{k_n +2M-2}{|\theta_{1i}- \theta_{2i}|}$. 
Using the above observations, we prove the Lemma.

Let $I_1=\int_0^1 g_1^{\omega} \circ \gamma_1 dt$ and $I_2=\int_0^1 g_2^{\omega} \circ \gamma_2 dt$. Then 
we have   $ 0 < r_n^l = \mbox{min}(\inf_i \lambda_{ki}, g_1^{\omega}(0),g_1^{\omega}(1))< I_k<\mbox{max}(1,\sup_i \lambda_{ki})=r_n^u $ for $k=1,2$.
Now, we have 
\begin{eqnarray*}
|f_1 - f_2| =\abs{\frac{(g_1^{\omega}\circ \gamma_1)I_1 - (g_2^{\omega}\circ \gamma_2)I_2}{I_1I_2}}=\abs{\frac{(g_1^{\omega}\circ \gamma_1)I_1 - (g_2^{\omega} \circ \gamma_2)I_1}{I_1I_2} + \frac{(g_2^{\omega}\circ \lambda_2)(I_1 - I_2)}{I_1I_2}}.
\end{eqnarray*}
Hence, 
\begin{eqnarray*}
|f_1 - f_2| \leq \abs{\frac{(g_1^{\omega}\circ \gamma_1) - (g_2^{\omega} \circ \gamma_2)}{I_2}} + \frac{(g_2^{\omega} \circ \lambda_2)}{I_1I_2}\abs{I_1 - I_2} \leq M_1  \sum_{i=1}^{k_n +2M-2}{|\theta_{1i}- \theta_{2i}|} + \frac{(g_2^{\omega} \circ \lambda_2)}{I_1I_2}\abs{I_1 - I_2}
\end{eqnarray*}
where the last inequality is obtained using the fact that $I_2$ is a finite positive number.
Now, $(g_2^{\omega} \circ \lambda_2)< \mbox{max}(1,r_n^u)$. Thus $(g_2^{\omega} \circ \lambda_2)/I_1I_2$ is bounded above by $r_n^{-2l} \mbox{max}(1,r_n^u)$. Next, it is easy to check that $ \abs{I_1 - I_2} \leq M_1{\|(g_1^{\omega}\circ \gamma_1) - (g_2^{\omega} \circ \gamma_2)\|}_{\infty} \leq M_2{\|(g_1^{\omega}\circ \gamma_1) - (g_2^{\omega} \circ \gamma_2)\|}_{1}$. Thus we have $|f_1 - f_2|\leq M_0 \sum_{i=1}^{k_n +2M-2}{|\theta_{1i}- \theta_{2i}|}$. 
\end{proof}

\begin{rem}
It follows that $H(f_1,f_2) <l_1\sqrt{{\|f_1-f_2\|}_1} <l_1\sqrt{ \sum_{i=1}^{k_n +2M-2}{|\theta_{1i}- \theta_{2i}|} }<  \linebreak 
l_1\sqrt{ \max_{1\leq j \leq k_n+2M-2} \abs{\theta_{1j} - \theta_{2j}}}$ for some fixed $l_1>0$ where $H(f_1,f_2)$ is the Hellinger metric between two densities $f_1$ and $f_2$.
\end{rem}

\begin{cor}
Let $p_0$ be the true density. If $k_n \sim n^{1/(2\beta+1)}$, then asymptotically $ \underset{f \in \mathcal{F}_n}{\inf} {\|p_0 - f\|}_{\infty} \sim n^{-\beta/(2\beta + 1)}$ where $\beta$ is the order of the Sobolev space.  
\end{cor}
This corollary follows from standard approximation results in $\ltwo$ basis (e.g. Fourier) of  H\"{o}lder functions of order $\beta$. For a detailed discussion please refer to \citet{triebel2006theory}. 

\begin{lem}
 There exists positive constants $C_3,C_4$, such that for some positive $\epsilon<1$, 
\begin{equation}
\int_{{\epsilon}^2/2^8}^{\sqrt{2}\epsilon} {H}^{1/2} (\frac{u}{C_3},\mathscr{F}_n) du \leq C_4n^{1/2}{\epsilon}^2
\label{hellinger}
\end{equation}
\end{lem}

\begin{proof}
The $u/C_3$-cover of a set $T$ with respect to a metric $\rho$ is a set $\{f^{1},\dots,f^{N}\}\subset T$ such that for each $f \in T$, there exists some $i\in \{1,\dots,N\}$ with $\rho(f,f^i)\leq u/C_3$. The covering number $N$ is the cardinality of the smallest delta cover. Then $\log(N)$ is the metric entropy for T.   First we  bound the metric entropy for $\mathscr{F}_n$.
Let us consider a fixed $f_1,f_2 \in \mathscr{F}_n$.
We choose the Hellinger metric for the space $\mathcal{F}_n$ so that we can  borrow
results directly from \cite{wong1995probability}. 
We note that  $H(f_1,f_2)\leq l_1 \sqrt{ \max_{1\leq j \leq k_n+2M-2} \abs{\theta_{1j} - \theta_{2j}}}$ for some $l_1>0$ following the Remark 1.
So finding a $u/C_3$ covering for $\mathscr{F}_n$ using Hellinger metric is equivalent to finding an $l_1 \sqrt{u/C_3}$ covering for the space of parameters $\Theta_n=\{\theta_n=(c,\lambda) |c \in {[-l_0,l_0]}^{k_n}, \lambda \in {(r_n^l,r_n^u]}^{2M-2}\} $ using $L_\infty$ norm for euclidean vectors.  The $l_1 \sqrt{u/C_3}$ covering number for $\Theta_n$ using $L_{\infty}$ norm is  ${(\frac{2l_0}{l_1} \sqrt{C_3/u})}^{k_n}{(\frac{(r_n^u - r_n^l)}{l_1} \sqrt{C_3/u})}^{(2M-2)}$. This is obtained by partitioning the intervals $[-l_0,l_0]$ and $[r_n^l,r_n^u]$ into pieces of length $l_1 \sqrt{u/C_3}$ corresponding to individual coordinates and thus obtaining the partition of $\Theta_n$ through cross product. Then in each equivalent class of the partition of $\Theta_n$ we have ${\|\theta_1 -\theta_2\|}_{\infty} \leq l_1  \sqrt{u/C_3}$. Thus the covering number is ${(\frac{2l_0}{l_1} \sqrt{C_3/u})}^{k_n}{(\frac{(r_n^u - r_n^l)}{l_1} \sqrt{C_3/u})}^{(2M-2)} < {(\frac{2l_0}{l_1} \sqrt{C_3/u})}^{k_n}{(\frac{r_n^u}{l_1} \sqrt{C_3/u})}^{(2M-2)}<{(\frac{2l_0\sqrt{C_3} + r_n^u \sqrt{C_3}}{l_1 \sqrt{u}})}^{(k_n + 2M-2)}=N$, say. So the metric entropy for $\mathcal{F}_n$, $H(u/C_3, \mathcal{F}_n)$ is bounded by $\log(N)=(k_n + 2M - 2)\log(\frac{2l_0\sqrt C_3 + r_n^u \sqrt C_3}{l_1\sqrt u})$.

Now, note that $r_n^u =\Omega_1 \log n$. Then there exists a constant $l_2$ such that $2l_0\sqrt C_3 + r_n^u \sqrt C_3 < l_2 r_n^u$. Also, let $k_n=n^{1/{(2\beta+1)}}=n^{\Delta}$. Then there exists a constant $l_3$ such that $k_n + 2M - 2< l_3k_n$.Thus we have, $\log(N) <l_3k_n \log(\frac{r_n^u l_2}{l_1 \sqrt u})$. Thus we have $H^{1/2}(u/C_3, \mathcal{F}_n)<\sqrt{\log N} < \sqrt{l_3k_n \log(\frac{r_n^u l_2}{l_1 \sqrt u})}$. Let $l_4=2^8l_2/l_1$. Hence,
\begin{eqnarray*}
\int_{{\epsilon}^2/2^8}^{\sqrt{2}\epsilon} H^{1/2}(u/C_3, \mathcal{F}_n) < \sqrt{l_3n^{\Delta}}\int \sqrt{\log \frac{l_2 r_n^u}{l_1\sqrt u}} < \sqrt{l_3n^{\Delta} \log \frac{l_4 r_n^u }{{\epsilon}^2}}(\sqrt{2}\epsilon - \frac{{\epsilon}^2}{2^8}) <\sqrt{2l_3{\epsilon}^2n^{\Delta} \log \frac{l_4 r_n^u }{{\epsilon}^2}}
\end{eqnarray*}
 Then as $\epsilon \uparrow 1$, there exists a constant $C_4$ such that $\sqrt{2l_3{\epsilon}^2n^{\Delta} \log \frac{l_4 r_n^u }{{\epsilon}^2}} \leq C_4n^{1/2}{\epsilon}^2$. Thus there exists an $\epsilon <1$ for which \eqref{hellinger} holds. 
\end{proof}

Now we are ready to provide the proof of Theorem \ref{thm2}.
\begin{proof}
Theorem $1$ of \citet{wong1995probability} states that, if  \eqref{hellinger} holds for some $\epsilon<1$, then there exists constants $C_1,C_2$ such that the following  likelihood surface inequality holds.
\begin{equation}
P^{*}\bigg( \underset{\{{\|p^{1/2}-p_{0}^{1/2}\|}_2 \geq \epsilon , p\in \mathscr{F}_n\}}{\text{sup}} \prod_{i=1}^{n} p(Y_i)/p_0(Y_i) \geq \text{exp}(-C_1n{\epsilon}^2)\bigg)\leq 4 \text{ exp}(-C_2n{\epsilon}^2)
\label{eq:p2}
\end{equation}

Next we derive an expression for an upper bound of the smallest $\epsilon<1$ that satisfies \eqref{hellinger}. Let the smallest $\epsilon$, denoted by $\epsilon_n$ be of the form $\sqrt{l_4} n^{-\eta}{(\log n)}^{\nu}$. Then
$\log \frac{l_4 r_n^u}{{\epsilon_n}^2}  =\log n^{2\eta}{(\log n)}^{1-2\nu} = (2 \eta)\log n +(1-2\nu) \log \log n < (\delta + 2\eta) \log n$.
Thus an upper bound for $\epsilon_n$ can be obtained by solving
\[
\sqrt{2l_3l_4n^{-2\eta}{(\log n)}^{2\nu}n^{\Delta} (2\eta \log n + (1-2\nu) \log \log n) } = C_4n^{1/2}l_4n^{-2\eta}{(\log n)}^{2\nu}.
\]
Setting $\nu =1/2$, and noting that $\Delta = 1/(2\beta + 1)$ we get $\eta=\beta/(2\beta + 1)$. Thus, $\epsilon_n = \sqrt{l_4} n^{\frac{-\beta}{2\beta + 1}}\sqrt{\log n}$ is an upper bound of the smallest $\epsilon$ that satisfies \eqref{hellinger}.

 Consider the family of discrepancies $\delta_n(\alpha)$ with $\alpha=1$. Let the true density be $p_0$ with corresponding parameters ${c_0}$ and $\lambda_0$.
$\delta_n(1)= \underset{p \in \mathscr{F}_n}{\inf} \rho_{1} (p_0,p).=\underset{p \in \mathscr{F}_n}{\inf} \int {(p_0-p)}^2/p$.
Let $p_1= \underset{ p \in \mathscr{F}_n}{\mbox{arginf}} \int {(p_0-p)}^2/p$.
Then $\delta_n(1) < {\|p_0-p_1\|}_{\infty}^2\int 1/f < {\|p_0 -p_1\|}_{\infty}^2 \min{(r_n^l,\omega)} \sim n^{-2\beta/{(2\beta+1)}}\log n$. 
Let $C_1,C_2$ satisfy  \eqref{eq:p2}.
 Define as in  Theorem 4 of \cite{wong1995probability}, 
\[
\epsilon_n^*  =\left\{\begin{array}{lr}\epsilon_n, \text{   if  } \delta_n (1) < \frac{1}{4}C_1{\epsilon_n}^2,\\
{(4\delta_n(1)/C_1)}^{1/2},\text{ otherwise.}
\end{array}
\right.
\]
Note that $\delta(1)$ and $\epsilon_n$ are equal upto constants.
It follows from Theorem 4 of \cite{wong1995probability},  that 
\begin{eqnarray*}\label{eq:main}
 P({\|{\hat{p}}^{1/2}-p_{0}^{1/2}\|}_2 \geq \epsilon_n^* ) \leq 5\exp \big\{-C_2n{(\epsilon_n^* )}^2\big\} + \exp \bigg\{-\frac{1}{4}nC_1{(\epsilon_n^* )}^2\bigg\}.  
\end{eqnarray*}
\end{proof}

\bibliographystyle{plainnat}
\bibliography{biblio}

\begin{thebibliography}{23}
\providecommand{\natexlab}[1]{#1}
\providecommand{\url}[1]{\texttt{#1}}
\expandafter\ifx\csname urlstyle\endcsname\relax
  \providecommand{\doi}[1]{doi: #1}\else
  \providecommand{\doi}{doi: \begingroup \urlstyle{rm}\Url}\fi

\bibitem[Abramson(1982)]{abramson1982bandwidth}
Ian~S Abramson.
\newblock On bandwidth variation in kernel estimates-a square root law.
\newblock \emph{The annals of Statistics}, pages 1217--1223, 1982.

\bibitem[Bashtannyk and Hyndman(2001)]{bashtannyk2001bandwidth}
David~M Bashtannyk and Rob~J Hyndman.
\newblock Bandwidth selection for kernel conditional density estimation.
\newblock \emph{Computational Statistics \& Data Analysis}, 36\penalty0
  (3):\penalty0 279--298, 2001.

\bibitem[Bickel and Fan(1996)]{Bickel1996-oi}
Peter~J Bickel and Jianqing Fan.
\newblock Some problems on the estimation of unimodal densities.
\newblock \emph{Stat. Sin.}, 6\penalty0 (1):\penalty0 23--45, 1996.

\bibitem[Birge(1997)]{Birge1997-gh}
Lucien Birge.
\newblock Estimation of unimodal densities without smoothness assumptions.
\newblock \emph{Ann. Stat.}, 25\penalty0 (3):\penalty0 970--981, 1997.

\bibitem[Brunner and Lo(1989)]{Brunner1989-rj}
Lawrence~J Brunner and Albert~Y Lo.
\newblock Bayes methods for a symmetric unimodal density and its mode.
\newblock \emph{Ann. Stat.}, 17\penalty0 (4):\penalty0 1550--1566, 1989.

\bibitem[Cheng et~al.(1999)Cheng, Gasser, and Hall]{Cheng1999-up}
Ming-Yen Cheng, Theo Gasser, and Peter Hall.
\newblock Nonparametric density estimation under unimodality and monotonicity
  constraints.
\newblock \emph{J. Comput. Graph. Stat.}, 8\penalty0 (1):\penalty0 1--21, 1999.

\bibitem[Dasgupta et~al.(2017)Dasgupta, Pati, and
  Srivastava]{dasgupta2017geometric}
Sutanoy Dasgupta, Debdeep Pati, and Anuj Srivastava.
\newblock A geometric framework for density modeling.
\newblock \emph{arXiv preprint arXiv:1701.05656}, 2017.

\bibitem[Einbeck and Tutz(2006)]{einbeck2006modelling}
Jochen Einbeck and Gerhard Tutz.
\newblock Modelling beyond regression functions: an application of multimodal
  regression to speed--flow data.
\newblock \emph{Journal of the Royal Statistical Society: Series C (Applied
  Statistics)}, 55\penalty0 (4):\penalty0 461--475, 2006.

\bibitem[Grenander(1956)]{Grenander1956-nt}
U~Grenander.
\newblock On the theory of mortality measurement: part ii.
\newblock \emph{Scand. Actuar. J.}, 1956.

\bibitem[Hall and Huang(2002)]{Hall2002-lw}
Peter Hall and Li-Shan Huang.
\newblock Unimodal density estimation using kernel methods.
\newblock \emph{Stat. Sin.}, 12\penalty0 (4):\penalty0 965--990, 2002.

\bibitem[Izenman(1991)]{Izenman1991-nk}
Alan~Julian Izenman.
\newblock Review papers: Recent developments in nonparametric density
  estimation.
\newblock \emph{J. Am. Stat. Assoc.}, 86\penalty0 (413):\penalty0 205--224,
  1991.

\bibitem[Lang(2012)]{lang2012fundamentals}
Serge Lang.
\newblock \emph{Fundamentals of differential geometry}, volume 191.
\newblock Springer Science \& Business Media, 2012.

\bibitem[Meyer(2001)]{Meyer2001-zh}
Mary~C Meyer.
\newblock An alternative unimodal density estimator with a consistent estimate
  of the mode.
\newblock \emph{Stat. Sin.}, 11\penalty0 (4):\penalty0 1159--1174, 2001.

\bibitem[Rao(1969)]{Rao1969-fc}
B~L S~Prakasa Rao.
\newblock Estimation of a unimodal density.
\newblock \emph{Sankhy\={a}: The Indian Journal of Statistics, Series A
  (1961-2002)}, 31\penalty0 (1):\penalty0 23--36, 1969.

\bibitem[Srivastava and Klassen(2016)]{srivastava2016functional}
Anuj Srivastava and Eric~P Klassen.
\newblock \emph{Functional and shape data analysis}.
\newblock Springer, 2016.

\bibitem[Terrell and Scott(1992)]{terrell1992variable}
George~R Terrell and David~W Scott.
\newblock Variable kernel density estimation.
\newblock \emph{The Annals of Statistics}, pages 1236--1265, 1992.

\bibitem[Triebel(2006)]{triebel2006theory}
Hans Triebel.
\newblock Theory of function spaces. iii, volume 100 of monographs in
  mathematics.
\newblock \emph{BirkhauserVerlag, Basel}, 2006.

\bibitem[Turnbull and Ghosh(2014)]{Turnbull2014-ag}
Bradley~C Turnbull and Sujit~K Ghosh.
\newblock Unimodal density estimation using bernstein polynomials.
\newblock \emph{Comput. Stat. Data Anal.}, 72:\penalty0 13--29, 2014.

\bibitem[Van~Kerm et~al.(2003)]{van2003adaptive}
Philippe Van~Kerm et~al.
\newblock Adaptive kernel density estimation.
\newblock \emph{Stata Journal}, 3\penalty0 (2):\penalty0 148--156, 2003.

\bibitem[Wasserman and Lafferty(2006)]{wasserman2006rodeo}
Larry Wasserman and John~D Lafferty.
\newblock Rodeo: Sparse nonparametric regression in high dimensions.
\newblock In \emph{Advances in Neural Information Processing Systems}, pages
  707--714, 2006.

\bibitem[Wegman(1970)]{Wegman1970-ux}
Edward~J Wegman.
\newblock Maximum likelihood estimation of a unimodal density, {II}.
\newblock \emph{Ann. Math. Stat.}, 41\penalty0 (6):\penalty0 2169--2174, 1970.

\bibitem[Wheeler et~al.(2017)Wheeler, Dunson, and Herring]{wheeler2017bayesian}
MW~Wheeler, DB~Dunson, and AH~Herring.
\newblock Bayesian local extremum splines.
\newblock \emph{Biometrika}, 104\penalty0 (4):\penalty0 939--952, 2017.

\bibitem[Wong and Shen(1995)]{wong1995probability}
Wing~Hung Wong and Xiaotong Shen.
\newblock Probability inequalities for likelihood ratios and convergence rates
  of sieve mles.
\newblock \emph{The Annals of Statistics}, pages 339--362, 1995.

\end{thebibliography}
\end{document}